\documentclass[12pt]{article}
\usepackage{titling}
\usepackage[normalem]{ulem}
\newcommand{\cO}[1]{}
\usepackage[longnamesfirst]{natbib}
\usepackage{enumitem}

\usepackage{xcolor,hyperref}
\usepackage{amsthm}
\usepackage{comment}
\definecolor{darkblue}{rgb}{0.0,0.0,0.7}
\hypersetup{colorlinks,breaklinks,linkcolor=darkblue,urlcolor=darkblue,anchorcolor=darkblue,citecolor=darkblue}

\usepackage{geometry}
\usepackage{setspace}

\theoremstyle{plain}
\newtheorem{theorem}{Theorem}

\newtheorem{cor}{Corollary}

\newtheorem{definition}{Definition}

\newtheorem{lemma}{Lemma}

\newtheorem{proposition}{Proposition}

\theoremstyle{plain}
\newcommand{\thistheoremname}{}
\newtheorem{genericthm}[theorem]{\thistheoremname}

\usepackage{units}
\usepackage{bbm}
\usepackage{enumerate}
\usepackage{booktabs}
\usepackage{tabularx}
\usepackage{arydshln}
\usepackage{amssymb}
\usepackage{amsmath}
\allowdisplaybreaks[1]
\usepackage{xcolor}
\usepackage{graphicx}
\usepackage{subfigure}
\usepackage{tikz}
\usepackage{pgfplots}
\usepackage{caption}
\usetikzlibrary{decorations.pathreplacing,calligraphy}
\pgfplotsset{width=10cm,compat=1.9}

\newcommand{\pbar}{\overline{p}}
\newcommand{\pund}{\underline{p}}
\newcommand{\Tbar}{\overline{T}}
\newcommand{\Tund}{\underline{T}}

\newcommand{\Bun}{B^{\mathrm{un}}}
\newcommand{\dBun}{\dot{B}^{\mathrm{un}}}
\newcommand{\Nun}{N^{\mathrm{un}}}

\newcommand{\Dr}{\overline\Delta(r)}

\definecolor{green}{HTML}{30AE17}
\usetikzlibrary{graphs}

\title{Dynamic Disclosure with(out) Timestamps\thanks{
We thank Simon Board, Brendan Daley, Bart Lipman, Elliot Lipnowski, George Mailath, Harry Pei, Andy Skrzypacz, Curtis Taylor, and seminar audiences at Boston University, Duke University, Indiana University, 2026 AEA CSWEP Session on Economic Theory, Bonn Christmas Microeconomic Theory Conference, University of Georgia Microeconomic Theory Mini-Conference, IIOC 2026, Pennsylvania Economic Theory Conference, Stanford University SITE sessions on Dynamic Games, Contracts, and Markets and on The Economics of Transparency, and University of Miami Theory Day for helpful comments and conversations. Kolb acknowledges financial support from the Philpott Faculty Fellowship at Indiana University.}}
\author{Aaron Kolb\thanks{Kelley School of Business, Indiana University, \url{kolba@iu.edu}.} \and Beixi
Zhou\thanks{Department of Economics, University of Pittsburgh, \url{beixi.zhou@pitt.edu}.}}
\date{September 18, 2026}

\begin{document}

\maketitle

\begin{abstract}
We study how timestamps affect dynamic disclosure. At a random date, a sender privately obtains hard evidence about an evolving binary state and chooses when to disclose it. With timestamps, the unique equilibrium features immediate good-evidence disclosure and timestamp-dependent bad-evidence delay. Without timestamps, high priors relative to impatience generate a stock of undisclosed good evidence, which is stochastically purged before permanent transparency. High impatience yields immediate good-evidence disclosure, while bad evidence can be arbitrarily delayed and disclosed in bursts. Timestamps prevent pretending old good evidence is fresh and certify bad evidence is old, accelerating good-evidence disclosure and facilitating bad-evidence disclosure.

\bigskip



\end{abstract}

\newpage
\onehalfspacing

\section{Introduction}
A key feature of many evidentiary documents is the presence or absence of a timestamp: evidence may convey not only information about an underlying state, but also information about when that evidence was generated. Medical test results, bankruptcy filings, audit reports, and legal records such as signed contracts and digital forensic logs often bear verifiable dates and can reveal the age of the evidence. By contrast, many commonly used documents are not timestamped. For example, sensationalist news media sometimes reuse old photos to support current narratives when original dates are difficult or costly for the public to verify. In addition, evidence of demand may be effectively untimestamped when recorded transaction dates do not reveal when the underlying demand arose, such as in sales ``sandbagging,'' where ready-to-close deals are deliberately deferred to a later reporting period.

There has been growing attention in the literature to how the timing of disclosure affects beliefs and incentives in dynamic environments. However, existing work typically abstracts from whether disclosure also reveals when the underlying evidence was generated. It remains an open question whether the presence or absence of timestamps plays a role in strategic disclosure over time.

We address this question by introducing a simple continuous-time sender-receiver model of dynamic disclosure of hard evidence about an evolving state, without commitment. Specifically, there is a hidden state $\theta_t$ that evolves according to Markov switching dynamics between two states, $1$ (``good'') and $0$ (``bad''). A sender privately obtains one piece of hard evidence at some random, exponentially distributed date $\tau$ and can choose to disclose it (or not) at any time. We consider two distinct environments: \textit{timestamps} and \textit{no timestamps.} In the timestamps environment, upon disclosure, the receiver sees both $\theta_\tau$ and $\tau$; in the no-timestamps environment, the receiver only sees $\theta_\tau$ and must infer $\tau$ itself. The sender's flow payoff is the receiver's posterior belief that the state is $1$, and the sender discounts future flow payoffs at rate $r$.  
This simple environment generates rich equilibrium dynamics that differ qualitatively between the timestamps and no-timestamps environments.

In the timestamps environment, we show that there is a unique equilibrium, and good evidence is disclosed immediately while bad evidence is disclosed after a deterministic, timestamp-dependent delay. To understand the intuition, consider first the incentive to disclose good evidence. If the sender possesses good evidence $\theta_\tau=1$ at $\tau$ and discloses it immediately, the receiver's belief jumps to $1$ and decays according to the hidden Markov switching dynamics thereafter. If instead the sender were to delay until some time $t>\tau$, the receiver would retrospectively update about $\theta_\tau$, so her belief path after $t$ would be exactly the same as under immediate disclosure. In equilibrium, the only effect of delay is to lower the receiver's belief during the interval $[\tau,t)$, while evidence is being withheld.\footnote{For some non-equilibrium conjectures, delay could temporarily raise the receiver's belief; our formal argument rules these out.} Timestamps prevent the sender from pretending good evidence is fresh, and thus the sender cannot benefit from delaying disclosure of good evidence.

In addition to this implication for good evidence, timestamps allow the sender to prove that old bad evidence is in fact old. Therefore, when the sender discloses old bad evidence, beliefs update retrospectively downward to $0$, but drift up according to the hidden Markov switching after that. Furthermore, since evidence only arrives once, once bad evidence has been disclosed, the receiver no longer expects good evidence to be disclosed. This shuts down a ``no news is bad news'' effect, thereby raising the trajectory of  beliefs. 
For each bad evidence arrival date, there is a unique date at which the receiver's belief after disclosure of that evidence crosses the no-disclosure belief path from below. This is the optimal time to disclose bad evidence. Although the full no-disclosure path depends on the full disclosure strategy, which in turn depends on the no-disclosure path, creating an infinite-dimensional fixed point problem, we introduce an auxiliary belief that allows us to reduce the problem to a set of one-dimensional fixed point problems.

A property of the unique  equilibrium under timestamps is that all evidence is disclosed at the first instant it is myopically optimal to do so; hence, the equilibrium is independent of the sender's discount rate. Additionally, we show that the Minimum Principle \citep{acharya2011endogenous,guttman2014not} applies in this timestamped setting: at all times, the receiver's equilibrium belief conditional on no disclosure is the most pessimistic belief consistent with Bayes' rule among all feasible disclosure strategies. Neither of these properties generally holds outside of the timestamps environment.

In the no-timestamps environment, we first show that if the prior belief is sufficiently high relative to the sender's discount rate, equilibrium always gives rise to a stock, from the receiver's perspective, of undisclosed evidence of various dates. If the receiver conjectured immediate disclosure of good evidence, a sender with good evidence could disclose it after a brief delay and convince the receiver that the evidence is fresh. The sender would obtain a lower flow payoff during the delay, but a higher payoff at all future times. When the prior belief is above a threshold, this trade-off favors early delay.\footnote{In standard disclosure models, equilibria often feature cutoff strategies, with highest types disclosing immediately. The delayed disclosure of good evidence here does not contradict this theme, because the notion of ``type'' applies not only to whether evidence is good or bad but also its arrival date, and from this perspective, there is no highest type --- there are always future arrival dates for the sender to mimic.} We also show that there must be a finite date $T$ at which the stock of undisclosed evidence becomes empty; furthermore, under a mild refinement on off-path beliefs, it must empty continuously by $T$. In particular, there cannot be an abrupt disclosure of the entire remaining stock at time $T$. Intuitively, the receiver's belief would not jump all the way to $1$ due to the pooling of old evidence, and the sender could simply wait an instant longer to separate from pre-$T$ types. Hence, the stock must be continuously cleared; this rules out, for instance, a bang-bang strategy of delay followed immediately by transparency.

Having established these general results for the no-timestamps environment, we specialize to the case where the bad state is absorbing to cleanly characterize  strategic delay of good evidence disclosure. This condition ensures that bad evidence is never disclosed.\footnote{We show that our construction goes through with a nonabsorbing but sufficiently persistent bad state.} For high priors, we construct an equilibrium consisting of three distinct phases, which we label \textit{stockpiling}, \textit{purging}, and \textit{transparency}. In the stockpiling phase, the sender does not disclose good evidence. In the purging phase, the sender stochastically discloses old evidence at an intensity that depends on calendar time but not on the evidence's arrival date; since the arrival date is payoff irrelevant for the sender, indifference for one type implies indifference for all types. During purging, the receiver's belief over past arrival dates therefore maintains full support. Eventually, the disclosure intensity increases, which causes the stockpile to deplete, shifting the distribution of evidence to more recent dates. This means that disclosure is eventually interpreted more favorably, which sustains the sender's indifference to delay in the first place. By the end of this phase, the disclosure intensity explodes and the stockpile fully (but continuously) empties. At that point, the receiver's belief is sufficiently low that the sender is willing to disclose any good evidence immediately.

The sender's incentive to delay in the purging phase depends not only on the receiver's belief about the current state, but also on the receiver's beliefs about the sender's evidence, since the latter belief determines the receiver's belief after disclosure and the evolution of the receiver's belief conditional on no disclosure. Technically, this presents a challenge in that the receiver's current belief about the state is not the only state variable. Moreover, the receiver must form a belief distribution at each instant about the continuum of past evidence arrival dates. However, with age-independent disclosure rates, we show that these beliefs can be summarized by a few stock variables that track the joint probability that the sender has good or bad evidence and the current state is good or bad, conditional on no disclosure. We describe the rest of the construction in detail in Section \ref{subsec:construction}.

Finally, we analyze disclosure of bad evidence. We focus on an impatient sender, which sustains immediate disclosure of good evidence and streamlines comparison to the timestamps environment. Using a stationary prior for convenience, we construct an equilibrium in which good evidence is disclosed immediately, while bad evidence is disclosed with delay in isolated bursts. Intuitively, a transparency phase cannot exist for bad evidence, since disclosures after the first instant would be interpreted as fresh. The delays between bursts dampen the belief drop at the time of disclosure by causing the pool of undisclosed bad evidence to age, ensuring that disclosure at the next burst is optimal; these delays become arbitrarily large as the discount rate increases.

Our analysis has implications for receiver welfare and policy when we microfound the sender's flow payoff as the result of a population of receivers choosing actions to match the state and minimize a quadratic loss. Timestamps not only convey information directly, strictly benefiting receivers at times when multiple evidence arrival dates would otherwise pool, but also change strategic disclosure incentives.  Without timestamps, there can be delayed disclosure of good evidence, so institutions that require timestamps can help the receiver get information sooner. Moreover, timestamps can facilitate the disclosure of bad evidence (when the bad state is not absorbing) by eliminating equilibria in which bad evidence is disclosed after long delays or never.

The remainder of this section discusses related literature. Section \ref{sec:model} introduces the model. Sections \ref{sec:timestamps} and \ref{sec:no_timestamps} analyze the timestamps and no-timestamps environments, respectively, and present our main results. Section \ref{sec:discussion} discusses receiver welfare, an alternative equilibrium, and modeling assumptions, and concludes. The Appendix and Online Appendix contain proofs.

\subsection*{Related literature}
Our paper contributes to the literature on strategic disclosure of hard information, or evidence, pioneered by \cite{grossman_disclosure_1980}, \cite{grossman_informational_1981}, and \cite{milgrom_good_1981}. They establish an unraveling result, where essentially all sender types disclose  to separate themselves from lower types. \cite{dye1985disclosure} and \cite{jung1988disclosure} show that when there is uncertainty about whether the sender has information, unraveling may fail, as low types with evidence can refrain from disclosing it to mimic high types that do not have evidence. 

More recent literature has studied disclosure in dynamic settings. In \cite{acharya2011endogenous}, a sender chooses when to disclose information about a fixed state of the world when there is exogenous public news about that state. The authors show that exogenous news creates a real options problem for the sender: by waiting to disclose, the sender retains the option to not disclose in case very favorable news arrives.

In \cite{guttman2014not}, the sender can obtain  multiple pieces of evidence sequentially across two stages, but the state (firm value) is constant and independent of the evidence arrival times. The authors show that late disclosures are interpreted more favorably. \cite{antic_pei_2026} study selective disclosure in an overlapping generations model with multiple, noisy pieces of evidence. In their model, evidence does not carry a timestamp. Although the state is fixed, the absence of timestamps is relevant in that disclosed signals convey information about the history of evidence arrival and what signals could have been concealed. Our focus is different from these papers; evidence arrival time and the current state are intimately linked due to hidden evolution of the state. 
This allows us to study optimal disclosure timing even when there is only one piece of evidence.

\cite{gratton2018bombshell} study a sender who privately observes a binary state and chooses when to initiate public information flow, trading off the benefit of signaling confidence through earlier disclosure with the cost of longer-term scrutiny before a deadline. \cite{zhou2025optimal} allows the sender to choose both when to open and when to close a disclosure window; delaying the start allows the sender to privately learn, and the resulting information asymmetry determines the duration of disclosure. \cite{marinovic2016nonews} consider disclosure about an evolving binary Markov state under disclosure costs and litigation risk. In contrast to our model, inferences about evidence age do not arise in these papers.

\cite{kremer2024disclosing} study disclosure of evidence about an evolving state. In their model, evidence is effectively timestamped, but the sender must either disclose it immediately or never.
In our model, delayed disclosure is permitted and is central to the analysis: with timestamps, bad evidence is optimally disclosed only after a delay, while without timestamps, good evidence may be delayed in order to appear fresh. Thus, relative to \cite{kremer2024disclosing}, our paper studies strategic delay of disclosure and its interaction with the presence or absence of timestamps. A further difference relates to the threshold for disclosing evidence. In \cite{kremer2024disclosing}, this threshold lies below the current mean belief of the receiver. In our model, with timestamps, bad evidence is disclosed only once it is old enough that disclosure doesn't lower the receiver's belief about the current state. Furthermore, without timestamps, good evidence is sometimes withheld even though disclosure would immediately raise the receiver’s belief.

A separate literature studies dynamic disclosure or persuasion with commitment.  \cite{ely2017beeps} and \cite{renault2017optimal} study dynamic persuasion with commitment about an evolving state. In   \cite{ely2025feedback} a principal with commitment provides feedback about success to an agent; optimal feedback is bang-bang: first silence, then transparency. In our no-timestamps environment, a bang-bang policy cannot occur in equilibrium. \cite{knoepfle2024dynamic} study social learning where a designer commits to a dynamic disclosure policy about a fixed state to induce experimentation. Bad evidence is disclosed immediately because the absence of bad news is good news. The distinction between timestamps and no timestamps in our model arises precisely from the combination of the sender's lack of commitment power and the evolving state.

\section{Model}\label{sec:model}
The game is played between a sender and a receiver in continuous time over an infinite horizon. There is a hidden, binary underlying state $\theta_t\in \{0,1\}$ that evolves according to Markov switching; when the state is $i$, it switches to $j\neq i$ at constant exponential rate $\lambda_i$. We assume $\lambda_1>0$ and $\lambda_0\geq 0$. Neither the sender nor the receiver observes the realization of the state process, and they share a common prior over the initial state $p_0=\Pr(\theta_0=1)\in (0,1)$. 

At some random date $\tau$, the sender privately obtains hard evidence of the current state $\theta_\tau$. The receiver does not directly observe $\tau$ or $\theta_\tau$. We assume that $\tau$ is distributed exponentially with parameter $\mu$ and is independent of the state process $(\theta_t)_{t\geq 0}$. The sender can receive evidence at most once. The sender then chooses to disclose this evidence at any date $\sigma\geq \tau$ or never disclose it, which we denote by $\sigma=\infty$. Disclosure effectively ends the game, although flow payoffs continue as described below. We consider two environments: \textit{timestamps} and \textit{no timestamps}. In the timestamps environment, when evidence is disclosed, the receiver observes both the realization $\theta_\tau$ and the timestamp $\tau$. In the no-timestamps environment, the receiver only observes the realization $\theta_\tau$ (and must infer $\tau$). Equivalently, evidence age is verifiable in the former environment but not the latter.  

A strategy for the sender is a collection of cumulative distribution functions $H(\cdot;\theta_\tau,\tau,t)$ on $[t,\infty]$, for $\theta_\tau\in \{0,1\}$ and $0\leq \tau\leq t$, where $H(s;\theta_\tau,\tau,t)$ is the probability that a sender who obtained evidence $\theta_\tau$ at time $\tau$ and has not disclosed it before $t$ plans to disclose by $s$. Strategies must satisfy the usual internal consistency requirement that $H(s;x,\tau,t')=\frac{H(s;x,\tau,t)-H(t'-;x,\tau,t)}{1-H(t'-;x,\tau,t)}$ whenever $s\geq t'>t$ and $H(t'-;x,\tau,t)<1$, where $t'-$ denotes the left limit at $t'$. We say that a strategy is \textit{age-independent} if for all $\tau,\tau',t$ with $0\leq \tau,\tau' \leq t$ and $x\in \{0,1\}$, $H(\cdot;x,\tau,t)=H(\cdot;x,\tau',t)$; that is, the continuation strategy at each time is independent of evidence age. We primarily focus on age-independent equilibria in the no-timestamps environment, where the timestamp is payoff-irrelevant to the sender.

A belief system for the receiver specifies at each time $t$, given the public history, a joint distribution over the current state $\theta_t$ and the sender's evidence history $E_t$, where $E_t=\emptyset$ denotes no evidence has arrived and $E_t=(\tau,x)\in [0,t]\times\{0,1\}$ if evidence arrived at time $\tau$ and has type  $x=\theta_\tau$.  We denote by $\pi_t$ the receiver's posterior belief at time $t$ that $\theta_t=1$; the process $\pi=(\pi_t)_{t\geq 0}$ is $[0,1]$-valued and c\`adl\`ag.\footnote{Note that $\pi_0=p_0$, but we use $\pi$ to distinguish the realized belief process from $p$, which we will define as the belief path conditional on no disclosure and which plays a central role in the analysis.} The sender's objective is to maximize
\begin{align*}
    \mathbb{E}\left[\int_0^\infty e^{-rt}\pi_t\,dt\right]
\end{align*}
where $r>0$ is the discount rate. Note that our receiver is passive, simply performing Bayesian updating; in Section \ref{sec:discussion}, we provide a microfoundation and discuss receiver welfare.

A perfect Bayesian equilibrium (henceforth, equilibrium) is a strategy and belief system such that the sender's continuation strategy $H(\cdot;x,\tau,t)$ is sequentially rational for all $(x,\tau,t)$ with $t\geq \tau$, and beliefs are consistent with Bayes' rule whenever possible.
Since evidence arrives stochastically, nondisclosure is always on path. Disclosure can be off-path. In the timestamps case, the receiver's belief after off-path disclosure is already determined by the evidence type and its timestamp. In the no-timestamps case, there is flexibility.\footnote{Recall that the sender does not observe the state evolution, and therefore the receiver's off-path beliefs are pinned down by her beliefs about the evidence arrival date; with timestamps, the latter is known.}

\paragraph{Belief evolution.} Given the continuous-time Markov chain, two belief processes will be useful throughout. The first is what we refer to as the no-information belief:
$$\phi_t := \Pr(\theta_t=1)=p^* + (p_0 - p^*) e^{-(\lambda_0 + \lambda_1) t},$$
where $p^*:=\frac{\lambda_0}{\lambda_0+\lambda_1}$ is the long-run steady-state belief. 
The no-information belief $\phi_t$ converges monotonically to $p^*$ from above (below) when $p_0 > p^*$ ($p_0 < p^*$). Second, adopting the convention that $G$ (``good'') refers to state $1$ and $B$ (``bad'') to state $0$, define $q^{G,s}_t := \Pr(\theta_t = 1 \mid \theta_s = 1)$ and $q^{B,s}_t := \Pr(\theta_t = 1 \mid \theta_s = 0)$ as the probabilities that $\theta_t = 1$ conditional on the state at $s \leq t$ being good or bad, respectively. Then for all $t \geq s \ge 0$ and $i\in\{G,B\}$, beliefs evolve according to
$$\dot{q}^{i,s}_t = \lambda_0 (1 - q^{i,s}_t) - \lambda_1 q^{i,s}_t$$
with initial conditions $q^{G,s}_s = 1$ and $q^{B,s}_s = 0$. These admit closed-form solutions 
\begin{align}
q^{G,s}_t = p^* + (1-p^*) e^{-(\lambda_0 + \lambda_1)(t - s)} \quad \text{and} \quad q^{B,s}_t = p^*- p^* e^{-(\lambda_0 + \lambda_1)(t - s)}.\label{qG and qB}
\end{align}
As $t\to \infty$, $q^{G,s}_t$ converges to $p^*$ from above and $q^{B,s}_t$ converges to $p^*$ from below.

Two particular cases are useful for our analysis. First, when $\lambda_0=0$, the bad state is absorbing. This corresponds to settings where, for example, the state is an attribute such as productivity or demand that can permanently deteriorate. In this case, $\phi_t$ converges to $p^*=0$ and $q_t^{B,s}=0$ for all $s$ and all $t\ge s$. An implication is that bad evidence disclosure is never optimal, with or without timestamps. This case allows us to isolate the role of timestamps in strategic disclosure of good evidence. Second, for any $\lambda_0$, when $p_0=p^*$, which we refer to as the stationary prior case, the no-information belief is constant. This shuts down the effect of the prior and allows us to isolate dynamics arising from stochastic evidence arrival and strategic disclosure alone. 
\section{Timestamps}\label{sec:timestamps}

\subsection{Equilibrium characterization}
In this section, we show that there exists a unique equilibrium under timestamps: good evidence is disclosed immediately, while bad evidence is disclosed after a deterministic, timestamp-dependent delay. This result is formalized in Theorem \ref{thm:timestamp_eq}.

For intuition, suppose the sender obtains good evidence at time $s$. The sender could choose to disclose it immediately and raise the receiver's belief immediately to $1$, after which it would decay toward the steady state due to the hidden Markov switching. Suppose instead the sender chooses to disclose at time $t>s$. Importantly, since the evidence is timestamped, the receiver learns that the state was $1$ at $s$, and the resulting belief is the same for all $u\geq t$ as it would have been after disclosure at time $s$. Moreover, the sender would obtain a lower flow payoff on $[s,t)$: under the conjectured equilibrium strategy, no news is weakly bad news, so the receiver's nondisclosure belief would be weakly lower than the belief given disclosure at time $s$. A subtlety is that immediate disclosure of good evidence is not a best response for \textit{all} conjectures of the sender's strategy: there exist conjectures for which the ranking of beliefs on $[s,t)$ need not hold. However, these conjectures would involve disclosure of bad evidence at a time when it would be suboptimal; the proof of Theorem \ref{thm:timestamp_eq} rules them out.

To understand why bad evidence must be disclosed eventually when the bad state is not absorbing, consider an equilibrium with immediate disclosure of good evidence but \textit{no} disclosure of bad evidence. Then nondisclosure reveals that either the sender has no evidence or has bad evidence: no news is bad news relative to the no-information belief path. A sender with old bad evidence would eventually prefer to disclose that evidence to separate himself from types with more recent bad evidence and to shut down the no news is bad news effect; in doing so, he switches the receiver to a higher belief path. 

This logic generalizes: as long as the bad state is not absorbing, for each timestamp $s$, the sender eventually finds it worthwhile to disclose bad evidence at the time $t$ when
\begin{align}
q^{B,s}_t=p_t,\label{eq:bad_ev_crossing_cond}
\end{align}
where the right-hand side of \eqref{eq:bad_ev_crossing_cond} is the equilibrium no-disclosure belief. As we show, for all $s$, $q^{B,s}_t$ crosses $p_t$ from below and only once, so disclosing at the crossing time is optimal independent of the sender's discount rate; the sender does not face an intertemporal trade-off.

Since the full path $(p_t)_{t\geq 0}$ is an endogenous object that depends on the exact disclosure time for each piece of evidence, the equilibrium characterization a priori involves an infinite-dimensional fixed point. However, by constructing a useful auxiliary belief, the problem reduces to a continuum of independent one-dimensional equations. To that end, we prove and exploit the following strict monotonicity property: at the instant each piece of bad evidence is disclosed, all earlier bad evidence will already have been disclosed, and no later bad evidence will have been disclosed. For each fixed $s$ and all $t\geq s$, define 
\begin{align}\label{eq:p_dagger}
    p^{\dagger}_{t,s}:=\frac{e^{-\mu t}\phi_{t}+\int_{s}^{t}\mu e^{-\mu\tau}(1-\phi_{\tau})q^{B,\tau}_td\tau}{e^{-\mu t}+\int_{s}^{t}\mu e^{-\mu\tau}(1-\phi_{\tau})d\tau}
\end{align}
as the belief at time $t$, conditional on no disclosure, under a conjecture that all good evidence is disclosed immediately, that bad evidence with $\tau\in [0,s]$ would be disclosed by time $t$, and that no bad evidence with timestamp $\tau>s$ would be disclosed by time $t$.

For each fixed $s$, $(p^{\dagger}_{t,s})_{t\geq s}$ and $(p_t)_{t\geq s}$ need not coincide since they are based on different conjectures. However, the key property is that whenever the optimality condition \eqref{eq:bad_ev_crossing_cond} holds, they do coincide. Hence, instead of characterizing the full path $(p_t)_{t\geq 0}$ for an arbitrary conjecture before solving \eqref{eq:bad_ev_crossing_cond} for each $s$, we can solve the equation (in $t$)
\begin{align}
q^{B,s}_t=p^{\dagger}_{t,s}.\label{eq:aux_belief_crossing}
\end{align}
The proof of our result establishes a unique solution to \eqref{eq:aux_belief_crossing} for all $s\geq 0$, as long as $\lambda_0>0$.

Overall, timestamps prevent the sender from pretending old good evidence is fresh, and allow the sender to prove old bad evidence is old.

\begin{theorem}\label{thm:timestamp_eq}
    There exists a unique equilibrium. If the sender ever possesses good evidence (on or off path), he discloses it immediately. If $\lambda_0=0$, bad evidence is never disclosed. For all $\lambda_0>0$, if the sender obtains bad evidence at time $s$, he discloses it at time $s+a^*(s)$, where for all $s\ge0$, $a^{*}(s)\in(0,\infty)$
 is the unique positive solution to 
\begin{align*}
 q_{s+a^{*}(s)}^{B,s}=p_{s+a^{*}(s),s}^{\dagger}.
\end{align*}
If the sender still possesses this evidence after $s+a^*(s)$ (off path), he discloses it immediately. 
\end{theorem}

\begin{figure}[ht!]
\begin{centering}
\includegraphics[scale=0.95]{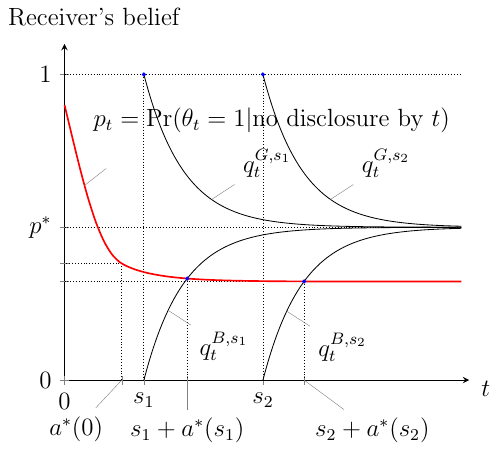}
\includegraphics[scale=0.95]{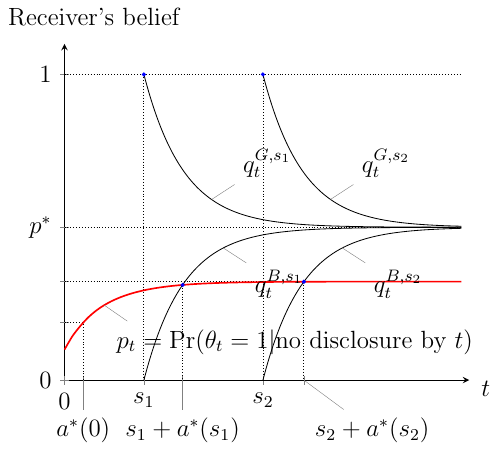}
\par\end{centering}

\caption{Equilibrium belief evolution for priors $p_{0}=0.9$ (left panel) and
$p_{0}=0.1$ (right panel), and parameters $\lambda_{0}=1,\lambda_{1}=1,\mu=0.5$. The equilibrium is independent of $r$.}\label{fig:equilibrium beliefs}
\end{figure}

Figure \ref{fig:equilibrium beliefs} illustrates the equilibrium belief dynamics under a high prior (left panel) and a low prior (right panel). If good evidence arrives at time $s$, where $s=s_1$ or $s_2$, the sender discloses it immediately and the receiver's belief jumps to $1$ and then follows the blue path. If bad evidence arrives at time $s$, the sender waits until the red post-disclosure belief path first reaches the black no-disclosure path. This crossing date is $s+a^*(s)$. Note that the equilibrium no-disclosure belief $p_t$ is piecewise. For $t<a^*(0)$, no bad evidence is disclosed, hence $p_t=p^\dagger_{t,0}.$
For $t\geq a^*(0)$, bad evidence with earlier timestamps has already been disclosed, whereas bad evidence with later timestamps is still being withheld. Hence, $p_t=p^{\dagger}_{s+a^*(s),s}$ where $t=s+a^*(s)$. The belief $p_t$ converges to a level strictly below $p^*$, the steady state of the Markov chain, a topic we explore in the next subsection.

\subsection{Equilibrium properties}\label{subsec:eqm properties}

We establish two sets of equilibrium properties. First, we examine how bad evidence disclosure timing varies with the evidence’s timestamp and how it shapes equilibrium beliefs.

\begin{proposition}\label{prop:timestamp properties}
 Fix $\lambda_0>0$. For all $s\ge0$, (i) $s+a^{*}(s)$
is increasing in $s$;
(ii) $a^{*}(s)$ is monotone in $s$: it is strictly increasing in $s$ if $p_0<p^*$, strictly decreasing if $p_0>p^*$, and constant if $p_0=p^*$; and (iii)  $\lim_{s\to\infty}a^{*}(s)=\alpha^{*}\in (0,\infty)$.
\end{proposition}

\begin{figure}[ht!]
\begin{centering}
\includegraphics[scale=1]{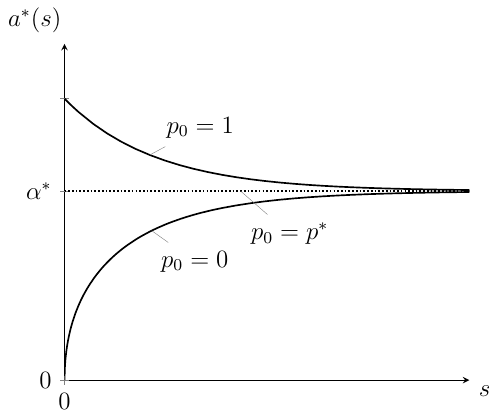}
\par\end{centering}
\caption{Optimal delay $a^*(s)$ as a function of bad evidence arrival time $s$ for priors $p_{0}=0,1,p^*$ and parameters $\lambda_{0}=1,\lambda_{1}=1,\mu=0.5$.}\label{fig:optimal delay}
\end{figure}
The fact that the optimal disclosure time $s+a^*(s)$ is increasing in the timestamp $s$ is due to the crossing condition \eqref{eq:bad_ev_crossing_cond} and the fact that older evidence is less damaging than new; hence, older evidence becomes worth disclosing first.

To understand why $a^*(s)$ is monotone, fix a delay $a$ and compare the disclosure condition at different calendar times. If bad evidence timestamped at $s$ is disclosed at $s+a$, its disclosure value depends only on the delay $a$ and not the calendar time because the Markov process is time-homogeneous. Calendar time therefore affects the disclosure condition only through the value of nondisclosure. For the marginal bad-evidence type, no disclosure pools no evidence with bad evidence that arrived within the last $a$ units of time. When $p_0<p^*$, the unconditional belief is increasing, so no evidence is better news at later dates, and evidence arriving during the later interval is less likely to be bad. No disclosure is therefore more valuable at later dates, so bad evidence must age longer for the belief upon disclosing it to catch up, so $a^*(s)$ is increasing. When $p_0>p^*$, the reverse holds. Figure \ref{fig:optimal delay} illustrates.

The behavior of $a^*(s)$ in turn determines the behavior of the equilibrium no-disclosure belief once bad evidence disclosure begins. As Figures \ref{fig:equilibrium beliefs} and  \ref{fig:optimal delay} illustrate, for $t\geq a^*(0)$, the no-disclosure belief moves in the same direction as $a^*(s)$. At the optimal disclosure date, the no-disclosure belief equals the posterior induced by the marginal bad evidence. When $p_0<p^*$, $a^*(s)$ is increasing, so the marginal bad evidence disclosed at later dates is older and therefore induces a higher post-disclosure belief. When $p_0>p^*$, $a^*(s)$ is decreasing, the reverse holds.\footnote{The claim that $p_t$ is monotone applies only after bad evidence disclosure begins, at $t\geq a^*(0)$; before $a^*(0)$, it may not be monotone.} In the long run, as $a^*(s)$ converges to $\alpha^*$, the age of the marginal bad evidence converges to $\alpha^*$, so the no-disclosure belief converges to the belief induced by bad evidence of exactly that age; since $\alpha^*<\infty$, this lies strictly below $p^*$. Intuitively, at large $t$, there is still a nontrivial probability that the sender received bad evidence within the last $\alpha^*$ units of time and has not yet disclosed it. The following corollary formalizes these results.

\begin{cor}\label{cor:no-disclosure-belief}
    Fix $\lambda_0>0$. At each $t\geq a^*(0)$, the equilibrium no-disclosure belief $p_t$ is strictly increasing in $t$ if $p_0<p^*$, strictly decreasing in $t$ if $p_0>p^*$, and constant in $t$ if $p_0=p^*$. Moreover, $\lim_{t\to\infty}p_t=p^*\left(1-e^{-(\lambda_0+\lambda_1)\alpha^*}\right)<p^*.$
\end{cor}

Next, we show that the equilibrium minimizes the receiver's belief conditional on nondisclosure at all times among all possible disclosure strategies.\footnote{Note that for any given conjectured strategy, the receiver's  belief conditional on nondisclosure is the same with or without timestamps.} This result extends the \textit{Minimum Principle} from earlier work of \cite{acharya2011endogenous} and \cite{guttman2014not}.\footnote{A related ``suspicious beliefs'' principle holds in \cite{kremer2024disclosing}.} Intuitively, for any fixed time $t$, the sender's strategy has a cutoff structure given by $s+a^*(s)=t$ (or $-\infty$ if there is no such $s$): all bad evidence with timestamp before this cutoff, along with all good evidence of any timestamp, is disclosed in equilibrium by time $t$ and would lead to a higher belief at $t$ than $p_t$; conversely, all bad evidence with timestamp after the cutoff would lead to a lower belief at $t$ than $p_t$. Hence, any alternative strategy would either involve nondisclosure of more favorable evidence or disclosure of less favorable evidence (or both). Both kinds of deviations would improve the nondisclosure pool in the sense of raising the average time-$t$ belief conditional on nondisclosure. 

\begin{proposition}[The Minimum Principle]
At each $t\geq 0$, the equilibrium no-disclosure belief $p_t$ is the minimum among all disclosure strategies.\label{prop:min_princ}
\end{proposition}

As we shall see next, the Minimum Principle does not generally hold without timestamps.

\section{No timestamps}\label{sec:no_timestamps}
Without timestamps, the sender can attempt to manipulate the receiver's beliefs about the age of evidence by strategically timing its disclosure, but the receiver's beliefs about evidence age are endogenous equilibrium objects. Since this behavior manifests differently for senders with good versus bad evidence, we divide our analysis into two parts.

\subsection{Good evidence: delay and gradual purging}
\label{subsec:no_timestamps_good}

We begin with a simple but general result showing that when the prior is sufficiently high relative to the sender's impatience, there \textit{must} be delayed disclosure of good evidence, creating a stock of undisclosed evidence from the receiver's perspective, in contrast to the timestamps environment.

\begin{proposition}[Immediate undisclosed stock]\label{prop:no_timestamps_necessary_delay}
If $p_0>\frac{\lambda_0+r}{\lambda_0+\lambda_1+r}$, then in any equilibrium, there exist arbitrarily small $t>0$ such that, conditional on no prior disclosure, there is strictly positive probability that the sender possesses undisclosed good evidence.
\end{proposition}

To understand Proposition \ref{prop:no_timestamps_necessary_delay} and the cutoff belief, suppose, to the contrary, that there is an equilibrium in which almost surely good evidence is disclosed immediately over some initial interval. If the sender obtains evidence at time $0$ and delays until some small time $\varepsilon>0$ (except on a set of measure $0$), the receiver interprets the evidence as fresh and updates to $1$. On $[0,\varepsilon)$, the posterior belief is lower than had the sender disclosed at time $0$. Since beliefs are right-continuous, to a first-order approximation, the loss from this delay is $(1-p_0)\varepsilon$. However, because of the delay, the belief at time $\varepsilon$ is higher by approximately $\lambda_1\varepsilon$, where $-\lambda_1$ is the slope of the post-disclosure belief at time $0$. Accounting for discounting and mean reversion after the disclosure, the net present value of this difference is $\lambda_1\varepsilon \frac{1}{\lambda_0+\lambda_1+r}$. Put together, delay is optimal if $(1-p_0)\varepsilon < \lambda_1\varepsilon \frac{1}{\lambda_0+\lambda_1+r}$, which rearranges to the condition in Proposition \ref{prop:no_timestamps_necessary_delay}. Intuitively, the trade-off favors delay when the prior is high. This delay is a strictly profitable deviation for almost all sufficiently early good evidence types.

Proposition \ref{prop:no_timestamps_necessary_delay} also implies that, when the prior is sufficiently high, the Minimum Principle fails, unlike in the timestamps environment. Delay means that the nondisclosure pool now contains some good evidence types, raising the nondisclosure belief. 

Our next general result states that in all equilibria, there exists $T>0$ at which the stock of undisclosed good evidence is empty. If not, the sender would be willing to conceal good evidence forever, and we show that this implies he would never disclose bad evidence. It follows that the no-disclosure belief would always lie weakly below the no-information belief, which lies strictly below the belief after disclosure of even time-$0$ good evidence, and there would be a strict incentive to deviate and disclose good evidence.

Furthermore, under a mild restriction on off-path beliefs, disclosure must involve a gradual, rather than abrupt, emptying of the stock. Let $S_t$ denote the probability that the sender has undisclosed good evidence at time $t$, conditional on no disclosure by $t$. We say that beliefs satisfy the \textit{support restriction} if whenever $S_T=0$ for some $T>0$, beliefs after off-path disclosure of good evidence at any $t>T$  place zero probability on arrival dates before $T$. Intuitively, after time $T$, there is zero probability that evidence that arrived before $T$ remains undisclosed, whereas new evidence can arrive after $T$.\footnote{Related support restrictions are used in multi-stage signaling games; see \cite{mcclellan2026signaling}. For a tremble-based motivation and further discussion, see the Online Appendix.} Now suppose, toward a contradiction, that good evidence is disclosed with an atom of probability at time $T$: $S_{T-}>0=S_T$. The receiver's belief given disclosure would jump to a value strictly less than $1$, based on her belief distribution over possible evidence ages. Under the support restriction, by delaying until $T+\varepsilon$, the sender can induce a posterior of at least $p^*+(1-p^*)e^{-(\lambda_0+\lambda_1)\varepsilon}$, since the receiver places probability one on arrival dates after $T$.\footnote{If $T$ is the beginning of a permanent transparency phase, then the result holds even without the support restriction, since then disclosure at $T+\varepsilon$ is on path and leads to a posterior of $1$.} This deviation yields a discrete benefit at arbitrarily small cost of delay, contradicting optimality of disclosure at the atom. In the equilibrium we construct, this dilemma is resolved through mixing.

\begin{lemma}[Eventual and continuous stock emptying]\label{lem:no_good_bang_bang} In any equilibrium, there exists $T>0$ such that the stock of undisclosed good evidence becomes empty: $S_T=0$. Moreover, under the support restriction, for all $T>0$ with $S_T=0$, we have $\lim_{t\uparrow T} S_t=0$.
\end{lemma}

Our game is a dynamic signaling game and therefore naturally admits multiple equilibria. Nevertheless, Proposition \ref{prop:no_timestamps_necessary_delay} and Lemma \ref{lem:no_good_bang_bang} impose structure on any equilibrium: for a high prior, an initial stock of undisclosed evidence develops, but it must eventually empty out, and (under the support restriction) must do so in a continuous fashion. 

Our goal therefore is not to characterize the full set of equilibria, but to obtain a sharp characterization within a natural class that makes clear the economics of strategic delay without timestamps. To that end, suppose now that $\lambda_0=0$; that is, the bad state is absorbing. (We return to $\lambda_0>0$ later.) Then disclosure of bad evidence is strictly dominated because it guarantees a flow payoff of $0$ forever after, isolating the role of good evidence.

We construct an equilibrium that exhibits continuous stock emptying via age-independent stochastic disclosure: disclosure behavior depends on calendar time but not the evidence's age. This class is natural because the actual evidence arrival date is not directly payoff relevant to the sender. We explore an age-dependent equilibrium in Section \ref{subsec:FIFO}. The following class of equilibria is central to our analysis. 

\begin{definition}\label{def:three_phase_eqbm}
    A three-phase equilibrium is characterized by two deterministic times $\Tund,\Tbar$ with $0\leq \Tund\leq \Tbar<\infty$ and a continuous function $\kappa_t>0$ on $(\Tund,\Tbar)$ such that 
    \begin{itemize}[leftmargin=1.25em,itemsep=.25pt,topsep=3pt] 
        \item \emph{Stockpiling phase:} On $[0,\Tund)$, the sender does not disclose good evidence. Following an off-path disclosure of good evidence, the receiver holds the \emph{neutral belief}: the probability that the current state is good conditional on the sender possessing good evidence.
         \item \emph{Purging phase:} On $[\Tund,\Tbar)$, good evidence disclosure is atomless from the receiver's perspective, and on $(\Tund,\Tbar)$, good evidence is disclosed at age-independent hazard rate $\kappa_t$.
        \item \emph{Transparency phase:} On $[\Tbar,\infty)$, the sender discloses good evidence immediately.
    \end{itemize}
\end{definition}

In particular, our stockpiling phase does not rely on belief threats to punish deviations; neutral beliefs are consistent with a post-disclosure posterior belief distribution over arrival dates identical to the pre-disclosure distribution, reflecting the fact that the arrival dates are not payoff relevant to the sender.\footnote{Neutral beliefs are the limit of strategies with age-independent trembles; that is, for each evidence arrival date $\tau \in [0,\underline T)$, the sender mixes on $[\tau,\underline T)$ by disclosing at constant hazard rate $\varepsilon$.}

We now present the main result of this section, which establishes a unique equilibrium within the three-phase class. Define $\pbar:=\frac{r+\lambda_{1}/2}{r+\lambda_1}$ and $\pund:=\frac{r}{r+\lambda_1}$. 

\begin{theorem}[Existence and uniqueness of three-phase equilibrium]\label{thm:three_phase_eqbm}
  Assume $\lambda_0=0$. For all $p_0,\lambda_1,\mu$ and $r$, there exists a unique three-phase equilibrium. Moreover,
  \begin{itemize}[leftmargin=1.25em,itemsep=.25pt,topsep=3pt] 
      \item If $p_0 >\pbar$, all three phases are nonempty.
      \item If $p_0 \in (\pund,\pbar]$, only the purging and transparency phases are nonempty.
      \item If $p_0\leq \pund$, only the transparency phase is nonempty.
  \end{itemize}
\end{theorem}

\begin{figure}[ht!]
    \centering
    \includegraphics[scale=1]{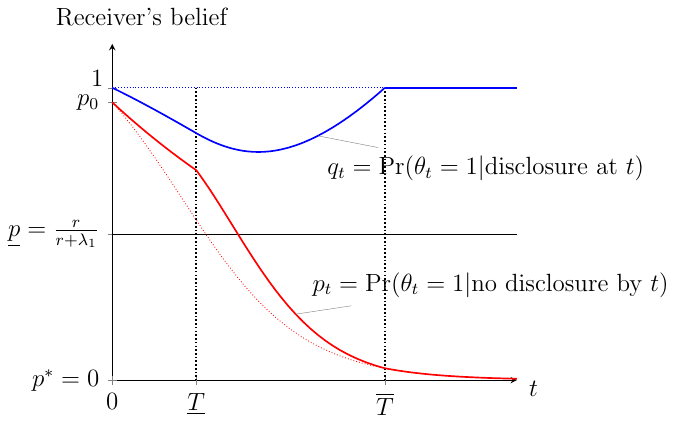}
    \caption{Belief paths in the three-phase equilibrium for $p_0>\pbar$, with vertical dashed lines at $\Tund$ and $\Tbar$. Parameter values are $\lambda_1=0.5,r=0.5$, $\mu=2$, $p_0=0.95$.}
    \label{fig:three_regions}
\end{figure}

Figure \ref{fig:three_regions} illustrates the equilibrium beliefs when $p_0>\pbar$; $p_t$ denotes the receiver's belief about $\theta_t$ conditional on no disclosure, and $q_t$ denotes the receiver's posterior immediately after disclosure of good evidence.\footnote{Note that while the belief jumps to $q_t$ at time $t$, after that, it evolves according to Markov switching; i.e., the receiver's belief does not continue to be given by $q$ at all future times after the disclosure.} For comparison, the dotted blue and red curves show the receiver’s beliefs conditional on disclosure and no disclosure, respectively, in the unique equilibrium with timestamps. As implied by the Minimum Principle, the no-disclosure belief without timestamps is bounded below by its counterpart with timestamps.

In the stockpiling phase, delay leads to an accumulation or stock of undisclosed evidence from the receiver's perspective. The receiver's belief $p$ conditional on no disclosure falls due to Markov switching, but beyond that, ``no news is no news'' as the receiver anticipates delay. Meanwhile, $q$ falls due to the Markov switching, but this fall is dampened by the arrival of new evidence. The local incentive compatibility condition for delay is
\begin{align}
    \dot{q}_t \geq r q_t -(\lambda_1+r)p_t.\label{eq:sec_4_delay_IC}
\end{align}
Intuitively, a marginal delay is beneficial when the current flow payoff $p_t$ is large, the posterior after disclosing good evidence $q_t$ is low, or when $\dot{q}_t$ is large (or not too negative). 

The purging phase begins at the time $\underline{T}$ when, if the receiver had counterfactually continued to conjecture delay, the sender would strictly prefer to disclose.\footnote{Recall that our three-phase equilibrium definition uses neutral off-path beliefs. Although an incentive to delay past $\underline T$ could be sustained for some time by more pessimistic off-path beliefs, the equilibrium conditions and atomless mixing immediately after $\underline T$ rule out longer delay. In other words, $\underline T$ is pinned down regardless of off-path beliefs, unless one goes beyond this class of equilibria to allow for atoms at the start of purging.} At the beginning of the purging phase, $p$ has a downward kink as the rate of disclosure is bounded away from $0$. During this phase, $q$ initially continues to fall as the rate of disclosure is small and does not offset the aging of the stock of undisclosed evidence. However, later in the purging phase, $p$ falls sufficiently low that indifference can only be sustained with an increasing $q$; this requires $\kappa$ to increase. As $q$ converges to its maximum value $1$, it must be that the composition of the pool of undisclosed evidence shifts toward fresh evidence; this requires the purging rate $\kappa$ to tend to infinity as $t\uparrow \Tbar$. When the transparency phase begins, it is no longer possible to maintain indifference, and the sender discloses good evidence immediately as it arrives. 

In Figure \ref{fig:three_regions}, it is noteworthy that the receiver's belief conditional on nondisclosure crosses $\pund$ before the end of the purging phase. In other words, the sender continues to (stochastically) delay even though the receiver's belief is sufficiently low that, had the game started at the same belief, there would be immediate disclosure at all times. This is because the incentive to delay depends not only on the current belief $p_t$ but also the post-disclosure value $q_t$ and its evolution. With a stockpile of undisclosed evidence, immediate disclosure does not convince the receiver that the current state is $1$, and delay can be weakly optimal. Thus, the variable $q$ is a critical additional factor in the sender's best response problem. 

\subsubsection{Three-phase equilibrium construction}\label{subsec:construction}
Recall that $q_t$ is the receiver's belief immediately after disclosure of good evidence. Clearly, $q_0=1$. Under neutral belief updating, an unexpected disclosure says that the sender has good evidence but the receiver infers nothing further about its age.
Rather than maintain a belief over the whole interval of past times, it is convenient to introduce ``stock'' variables
\begin{align*}
    G_{i,t}=\Pr(\tau \leq t \text{ and } \theta_\tau = 1 \text{ and } \theta_t=i | \text{no disclosure by $t$}) \qquad \text{for $i\in \{0,1\}$},
\end{align*}
representing the probability that the sender possesses good evidence \textit{and} the state is currently $i$, conditional on no disclosure yet. It follows that for $t\in (0,\underline T)$,
\begin{align}
    q_t=\frac{G_{1,t}}{G_{1,t}+G_{0,t}}.\label{eq:q_of_G}
\end{align}
The same characterization also applies after disclosure in the purging phase, since the disclosure rate of good evidence is independent of its timestamp.

In either the stockpiling (setting $\kappa=0$) or purging phase, $(G_1,G_0,p)$ evolve according to
 \begin{align}
    \dot{G}_{1,t}&=-\lambda_1 G_{1,t}+(p_t-G_{1,t})\mu -G_{1,t}\kappa_t(1-G_{1,t}-G_{0,t})\label{eq:G1_ODE_mix}\\
\dot{G}_{0,t}&=\lambda_1 G_{1,t}-G_{0,t}\kappa_t(1-G_{1,t}-G_{0,t})\label{eq:G0_ODE_mix}\\
    \dot{p}_t&=-\lambda_1 p_t-G_{1,t}\kappa_t+p_t(G_{0,t}+G_{1,t})\kappa_t,\label{eq:p_ODE_mix}
\end{align}
with initial values $(0,0,p_0)$. In the $G_1$ and $G_0$ ODEs, the $\lambda_1$ terms reflect that Markov switching from state $1$ to state $0$ increases $G_0$ while decreasing $G_1$. The $(p_t-G_{1,t})\mu$ term represents the arrival of good evidence, which occurs at rate $\mu$ when the state is good and the sender does not already possess good evidence. The $-\kappa_t G_{i,t}$ terms reflect outflows due to disclosure, and the adjustment factor $(1-G_{1,t}-G_{0,t})$ accounts for the conditioning on no disclosure yet. The $p$-ODE terms have an analogous interpretation.\footnote{When the bad state is not absorbing ($\lambda_0>0$), we also must introduce analogous variables $B_{i,t}$ corresponding to bad evidence. The variable $B_1$ enters the law of motion for $p$ when $\lambda_0>0$ since it is possible that the state is good but was bad in the past and the sender already obtained bad evidence, preventing the arrival of good evidence. When $\lambda_0=0$, $B_1$ is identically $0$.}

In the stockpiling phase the system has a closed-form solution. Provided $p_0>\pbar$, the local delay IC condition \eqref{eq:sec_4_delay_IC} holds strictly until some time $\Tund$, when the sender becomes locally indifferent. This marks the end of the stockpiling phase.

Throughout the purging phase, the IC condition \eqref{eq:sec_4_delay_IC} must hold with equality. 
The system in $(G_1,G_0,p,q,\kappa)$ is a mixed differential-algebraic system rather than a standard initial value problem (IVP). This is because $\kappa$ is not given exogenously, and instead we have an additional equation \eqref{eq:q_of_G} linking $q$ to $G_0$ and $G_1$. To solve the system, we eliminate $G_0$ and $G_1$ using $S:=G_0+G_1$ and $G_1=q S$, and then we reduce the system to an IVP in two variables $(p,q)$ by solving for $S$ and $\kappa$ as functions of $(p,q)$. The resulting equations are \eqref{eq:mix_IVP_p} and \eqref{eq:mix_IVP_q} in the Appendix, and the dynamics of this system are illustrated in Figure \ref{fig:phase_diagram}. The lightly shaded, outlined region, $R$, is where $S$ and $\kappa$ are well defined and $\kappa$ is positive (and therefore a valid hazard rate). In the proof, we show that the jump in the receiver's belief caused by disclosure at time $t$, $\delta_t:=q_t-p_t$, evolves as
    \begin{align}
\dot{\delta}_t=\delta_t(r+S_t\kappa_t),\label{eq:mix_IVP_delta}
    \end{align}
which implies that sustained indifference requires exponential growth of $q_t-p_t$. Hence, $(p,q)$ must exit this shaded region at some finite time $\Tbar$, and we show that it can only do so through $q_{\Tbar}=1$, marking the end of the purging phase. Hence, as $t\uparrow \Tbar$, the purging rate $\kappa_t$ explodes and the stockpile empties, i.e. $S_t\to 0$.

After the purging phase, the stockpile is empty and the receiver conjectures immediate disclosure at all times. Hence, $q_t=1$ for all $t\geq \Tbar$. Since the receiver's belief $p_t$ is already below the threshold $\pund$, immediate disclosure is optimal from then on. 

\begin{figure}[ht!]
    \centering
\includegraphics[width=0.5\linewidth]{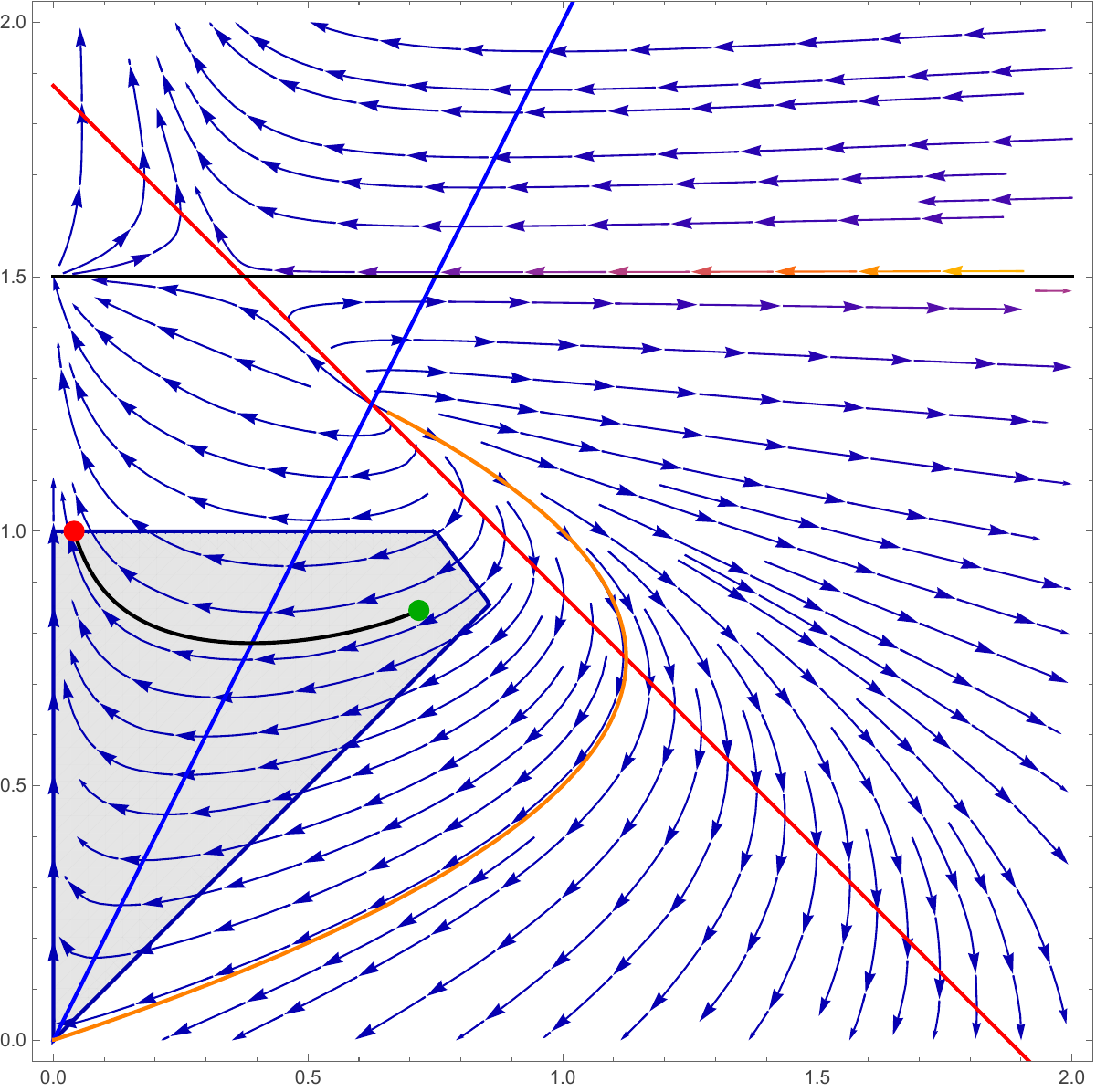}
    \caption{Phase diagram for reduced purging-phase system in $(p,q)$. Beliefs start at the green dot and the solution first exits $R$ through $q_{\Tbar}=1$ at the red dot. The red and blue lines are isoclines for $p$ and $q$, respectively, where their derivatives change sign. The ODEs are well defined below the black horizontal line. The orange curve is the set of initial conditions for which the solution converges to the absorbing steady state $(0,0)$. Parameter values are $\lambda_1=0.5,r=0.5$, $\mu=2$, $p_0=0.95$.}
    \label{fig:phase_diagram}
\end{figure}

\subsubsection{Robustness} 
Our restriction to $\lambda_0=0$ is useful as it isolates the analysis of good evidence disclosure by shutting down disclosure of bad evidence. That said, our results can be extended to $\lambda_0>0$. The following result gives a sufficient condition for there to continue to be an equilibrium with no disclosure of bad evidence.

\begin{proposition}\label{prop:three_phase_lambda0>0}
For all sufficiently small $\lambda_0>0$, there exists an equilibrium  with the same three-phase structure for good evidence and no disclosure of bad evidence, supported by beliefs that bad evidence disclosed off-path is fresh. 
\end{proposition}

\subsection{Bad evidence: delay and disclosure bursts}\label{subsec:no_timestamps_bad}
Throughout this section, we assume the bad state is not absorbing, $\lambda_0>0$, which is necessary for bad evidence disclosure to ever be optimal. We focus on equilibria with (i) age-independent disclosure of bad evidence, and (ii) immediate disclosure of good evidence. Property (i) is natural because timestamps are payoff-irrelevant to the sender, as in the three-phase equilibrium. Property (ii) serves two purposes: first, to isolate the effect of timestamps on bad evidence disclosure, as good evidence is also disclosed immediately in the timestamps environment, and second, to isolate the effect of parameter changes on bad evidence disclosure incentives within the no-timestamps environment by holding good evidence disclosure fixed. To accommodate (ii), we focus on large discount rates, under which, as we verify, immediate disclosure of good evidence is indeed optimal whenever it is conjectured.

In choosing whether and when to disclose bad evidence, the sender faces a trade-off between a long-run benefit and a short-run loss. The benefit is the same as with timestamps: because evidence arrives only once, disclosure eliminates the negative inference from not disclosing good evidence. Without timestamps, however, the sender cannot prove that old bad evidence is indeed old and wait until disclosure is myopically beneficial, so the receiver averages over possible arrival dates upon disclosure of bad evidence. With age-independent disclosure, disclosing bad evidence unambiguously and immediately decreases the receiver’s belief about the state. Bad evidence disclosure proves that the sender had bad evidence, and the receiver's distribution over arrival dates simply scales up proportionally to account for the mass that was previously assigned to the probability that no evidence has arrived. 

Thus, the sender's optimal disclosure timing depends critically on his discount rate. This is in stark contrast to the timestamps environment, where the equilibrium is independent of the discount rate. The following result says that this contrast can become extreme: by taking $r$ sufficiently large, bad evidence can only be disclosed after an arbitrarily long delay.

\begin{proposition}\label{prop:no_ts_long_delay_bad}
For any $T>0$, if $r$ is sufficiently large, then in all age-independent equilibria with immediate good evidence disclosure, bad evidence disclosure has zero probability in $[0,T]$.
\end{proposition}

It is worth noting that although bad evidence is withheld, the Minimum Principle (Proposition \ref{prop:min_princ}) says that the receiver is no more pessimistic following nondisclosure than in the timestamps environment. This may seem surprising at first because withholding pools bad-evidence types with the no-evidence type. The reason is that, with timestamps, older, and therefore more favorable, bad evidence is disclosed and removed from the nondisclosure pool; without timestamps, such older evidence remains in the pool and raises the average belief.

Proposition \ref{prop:no_ts_long_delay_bad} does not specify how or whether bad evidence is disclosed after time $T$. The answer depends on the receiver's belief following an off-path disclosure of bad evidence. Under belief threats that specify any such disclosure must be fresh, there exists an equilibrium where the sender never discloses bad evidence at any time if $r$ is sufficiently large. Under neutral beliefs,\footnote{Neutral belief is defined analogously to Definition \ref{def:three_phase_eqbm}. It is the probability the state is good conditional on the sender possessing bad evidence. } there is no such equilibrium: at some point, the sender strictly prefers disclosing to never disclosing. The following lemma formalizes this distinction.

\begin{lemma}\label{lem:bad_eventually_disclosed}
Under neutral beliefs after off-path disclosure of bad evidence, there is no equilibrium in which good evidence is immediately disclosed and bad evidence is never disclosed. Under threat beliefs, such an equilibrium exists if $r$ is sufficiently large.
\end{lemma}

Intuitively, following an off-path disclosure of bad evidence, neutral beliefs specify that the receiver averages over possible arrival dates. Over time, she assigns more weight to older, less informative evidence, so the neutral belief converges to $p^*$. The nondisclosure belief also converges to $p^*$, but remains higher because evidence may not have arrived. Thus, late disclosure causes little belief loss while eliminating future negative inference from the nondisclosure of good evidence, eventually making disclosure profitable.\footnote{Formally, both the benefit and loss vanish as $t\to\infty$, but the loss vanishes faster, as shown in the proof.}

Focusing on neutral beliefs, although Lemma \ref{lem:bad_eventually_disclosed} implies eventual disclosure of bad evidence, it does not specify how disclosure occurs. We first show that immediate disclosure of bad evidence from some point onward cannot be an equilibrium. If all bad evidence were disclosed immediately starting at some $T$, any disclosure after $T$ would be interpreted as fresh. The receiver's belief would drop to $0$ and then evolve according to the ``no news is no news'' Markov switching dynamics. By not disclosing, the receiver's belief stays positive and still evolves according to those same dynamics: the receiver expects both good and bad evidence to be disclosed immediately, nondisclosure means that none has arrived. Disclosure would result not merely in a lower expected payoff but a lower flow payoff at all future times.

\begin{lemma}\label{lem:no_bad_immediate}
    There is no equilibrium in which good evidence is always disclosed immediately and for some $T\geq 0$ any bad evidence possessed at $T$ or later is disclosed immediately.  
\end{lemma}

The preceding results point toward equilibria in which bad evidence is withheld initially but later released in an \textit{intermittent} fashion. Our main result for bad evidence, Theorem \ref{thm:stationary_bursts}, establishes the existence of an equilibrium with precisely these features: bad evidence is disclosed with delay along an infinite sequence of isolated \emph{bursts}. To simplify the characterization, we use a stationary prior $p_0=p^*$, which admits a periodic equilibrium.

\begin{theorem}[Bad evidence in isolated bursts]\label{thm:stationary_bursts}
    Assume $p_0=p^*$. Then if $r$ is sufficiently large, there exists an equilibrium in which good evidence is disclosed immediately and bad evidence is disclosed at the next time in the set $\{0,\Delta,2\Delta,\dots\}$, where $\Delta>0$. Neutral off-path beliefs are sufficient to support this equilibrium.
\end{theorem}

The equilibrium in Theorem \ref{thm:stationary_bursts} exhibits periodic behavior illustrated in Figure \ref{fig:equilibrium beliefs burst}. Given $p_0=p^*$, the no-information belief $\phi_t$ remains constant at $p^*$. In between bursts, only good evidence is disclosed, so no news is bad news, causing the no-disclosure belief to decay, but this decay can be offset eventually by upward mean-reversion due to the hidden Markov switching, resulting in the hook shape. Since each burst results in the disclosure of any possessed evidence, the belief resets to $p^*$. Meanwhile, the neutral off-path belief resets to $0$ after each burst due to the empty stockpile, and it rises as the stock of bad evidence ages.

The existence of such an equilibrium rests on two key incentive constraints. The first  is that in between bursts, if the sender has bad evidence, he must be willing to wait until the next burst; this constraint places an upper bound on $\Delta$, denoted $\Dr$. The second constraint is that the sender must be willing to ``participate'' in any given burst rather than wait until the next one; this constraint rules out sufficiently small $\Delta$, and the interval of such $\Delta$ becomes arbitrarily large as $r\to \infty$ by Proposition \ref{prop:no_ts_long_delay_bad}. However, as $r\to \infty$, we also have $\Dr\to\infty$. The proof shows that the two constraints are satisfied when $r$ is large by showing the participation constraint is slack at $\Dr$.

\begin{figure}[ht!]
\centering
\includegraphics[scale=1]{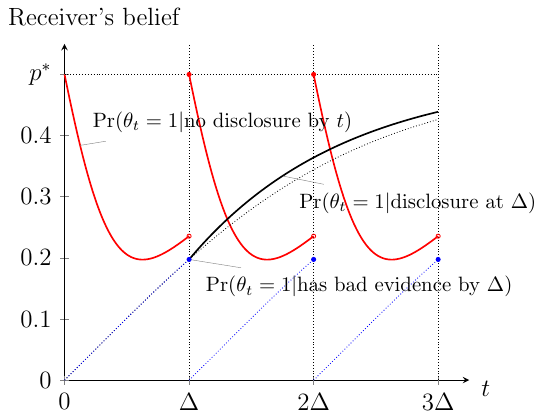}

\caption{The receiver’s beliefs in a periodic burst equilibrium. The red curve shows the no-disclosure belief, the dotted blue curves show off-path beliefs following disclosure between bursts, and the solid black curve shows the belief following disclosure at the burst time \(\Delta\). The dotted black curve gives the counterfactual continuation of the neutral belief after $\Delta$ (under a conjecture of no bad evidence disclosure) for contrast. The  parameters are $\lambda_{0}=1,\lambda_{1}=1,\mu=10,p_0=p^*=1/2$, $r=5$, and $\Delta=0.4$.
}
\label{fig:equilibrium beliefs burst}
\end{figure}

We make two remarks about the burst equilibrium construction. First, the stationary prior gives a clean characterization, but the intermittent burst phenomenon is not knife-edge. If $p_0\ne p^*$, the no-disclosure belief resets to $\phi_t$ at each burst, which is no longer constant, and thus  bursts need not be equally spaced. Nevertheless, we conjecture that the burst structure persists. Second, under neutral beliefs, any equilibrium with disclosure in isolated bursts must feature infinitely many bursts. Immediately after each burst, the evidence stock is empty and the game effectively restarts with a different prior. Lemma \ref{lem:bad_eventually_disclosed} then implies that no burst can be the last one.

In summary, with timestamps, every piece of bad evidence is disclosed in finite time, with timing independent of $r$. Without timestamps, disclosure can be postponed beyond any fixed horizon if $r$ is sufficiently large. The burst equilibrium illustrates a natural way disclosure may occur. Timestamps allow the sender to verify that bad evidence is old, whereas their absence can generate much longer delays and pooling across evidence ages.

\section{Discussion}\label{sec:discussion}
\subsection{Welfare}\label{subsec:welfare}
We first show that the sender’s ex ante expected payoff is the same under any disclosure strategy, as long as the receiver’s conjecture is correct. Thus, the sender cannot benefit from commitment. The lemma below follows directly from the law of iterated expectations.

\begin{lemma}\label{lem:sender_indifference}
    Fix parameter values. Fix any strategy of the sender, and let $(\pi_t)_{t\geq 0}$ denote the receiver's posterior belief process when that strategy is followed and the receiver's conjecture is correct. Then for all $t\geq 0$, $\mathbb{E}[\pi_t]=\phi_t=p^{*}+(p_0-p^{*})e^{-(\lambda_0+\lambda_1)t}$, which is independent of the sender's strategy. Thus, the sender's expected payoff is $\int_0^\infty e^{-rt}\phi_t\,dt$.
\end{lemma}

Although the sender is ex ante indifferent across equilibria, this is not generally true path-by-path. For instance, consider the timestamps equilibrium and the burst equilibrium in the no-timestamps environment. By Proposition \ref{prop:min_princ}, the sender's flow payoff prior to any disclosure is weakly higher in the burst equilibrium; furthermore, his payoff after good evidence disclosure is the same in both. The sender therefore weakly prefers the burst equilibrium on every path in which good evidence eventually arrives. Ex ante indifference then implies that, on average over paths in which bad evidence eventually arrives, he weakly prefers the timestamps equilibrium. By the same logic, the comparison is reversed for the three-phase equilibrium of Theorem \ref{thm:three_phase_eqbm} with the timestamps equilibrium (under $\lambda_0=0$): the sender weakly prefers the three-phase equilibrium on paths in which bad evidence arrives, and the timestamps equilibrium averaging over paths in which good evidence arrives.

The sender's flow payoff $\pi_t$ in our model is reduced form and can be microfounded as follows. Consider a market with a unit mass of receivers indexed by $i\in [0,1]$ who choose actions $a_{i,t}$ at each instant $t$ and receive (unobserved) payoffs $-\int_0^\infty e^{-\rho t} (a_{i,t}-\theta_t)^2\,dt$. The sender's flow payoff is $\int_0^1 a_{i,t}\,di$. Then $a_{i,t}=\pi_t$ for all $(i,t)$, and the sender's flow payoff is exactly $\pi_t$ as in our model. We then use ``receiver'' to refer to a representative receiver.

This microfoundation allows us to compare receiver welfare across environments. In the three-phase equilibrium, delayed disclosure of good evidence reduces receiver welfare relative to immediate disclosure under timestamps. And in our three-phase equilibrium extension to small $\lambda_0>0$, bad evidence is withheld permanently, whereas it is eventually disclosed with timestamps. In both cases, the timestamps equilibrium Blackwell-dominates the three-phase equilibrium. In the burst equilibrium, good evidence is disclosed immediately, but delayed disclosure of bad evidence reduces receiver welfare when the sender is sufficiently impatient. As the sender’s discount rate $r$ increases, bad evidence disclosure can be postponed beyond any fixed horizon, whereas its delay under timestamps remains finite and independent of the discount rate. Holding the receiver’s discount rate $\rho$ fixed, sufficiently large $r$ therefore reduces receiver welfare relative to timestamps.\footnote{This holds more generally beyond the burst equilibrium: for large $r$, the receiver is worse off in any age-independent equilibrium with immediate good evidence disclosure than in the timestamps equilibrium.} We state and prove this result formally in the Online Appendix. Thus, removing timestamps can reduce receiver welfare both directly, by hiding the evidence's age, and indirectly, by delaying or preventing disclosure.

\subsection{Age-dependent equilibrium}\label{subsec:FIFO}
In the no-timestamps environment, our three-phase equilibrium uses delay and then age-independent mixing to generate and empty the good evidence stockpile. An age-dependent disclosure strategy can also perform this task, illustrating the robustness of the phenomenon.

Define a ``first in, first out'' (``FIFO'') equilibrium as one in which bad evidence is never disclosed, and good evidence disclosure is characterized by a continuous, strictly increasing function $D$ on $[0,\infty)$ and a time $T^F\geq 0$ such that $D(t)>t$ for all $t\in (0,T^F)$ and $D(t)=t$ for all $t\in \{0\}\cup [T^F,\infty)$. The sender plays a pure, separating strategy, which effectively reveals the timestamp of each disclosure. Disclosure begins at time $0$, even when $p_0 > \pbar$, but an undisclosed stock still develops as evidence arrival outpaces disclosure. Eventually, disclosure accelerates and the stock is emptied by time $T^F$.\footnote{There cannot be a FIFO equilibrium variant with $D(0)>0$: the $0$ type would prefer to disclose at $0$.}

\begin{proposition}\label{prop:FIFO}
    Fix $\lambda_0=0$. For all $p_0>\pund$, there exists a unique FIFO equilibrium. Moreover, its transparency phase starts strictly later than in the three-phase equilibrium for all $p_0 >\overline p$, and at the same time when $p_0 \in (\underline p,\overline p]$.
\end{proposition}

The FIFO equilibrium admits a closed-form solution for $D(t)$. Nonetheless, we contend that the three-phase equilibrium is more focal for several reasons. First, FIFO requires breaking indifference over disclosure times in a highly specific timestamp-dependent way, despite the timestamp's payoff irrelevance to the sender. In contrast, the three-phase equilibrium naturally treats all evidence timestamps symmetrically. Second, delay during the stockpiling phase $[0,\underline T)$ is strictly optimal even under neutral off-path beliefs, whereas in the FIFO equilibrium delay is only weakly optimal, even under the (correct) receiver conjecture that the worst possible undisclosed evidence is being disclosed at each instant. 

Third, and perhaps most importantly, FIFO results in a weakly later start of transparency, and strictly so for high priors, as stated in Proposition \ref{prop:FIFO} and illustrated in Figure \ref{fig:FIFO_vs_3P}. Intuitively, for a fixed distribution of undisclosed evidence, the receiver's belief jumps higher after observing disclosure of a random piece of evidence rather than the oldest possible one. Maintaining sender indifference therefore requires faster improvement of $q_t$, which is accomplished by faster disclosure in the three-phase equilibrium.

\begin{figure}[ht!]
\centering
\includegraphics[scale=1]{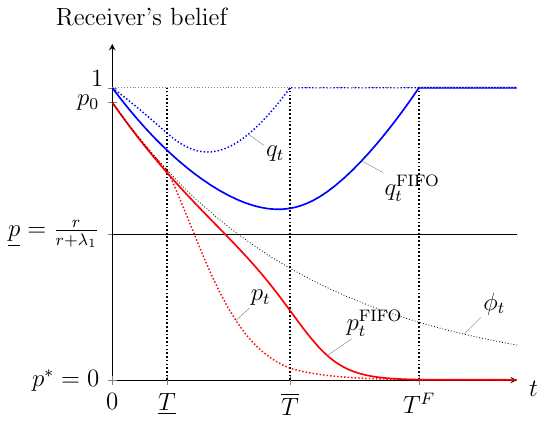}

\caption{The receiver’s beliefs in the FIFO equilibrium vs the three-phase equilibrium and the no-information benchmark. Parameter values are the same as in Figure \ref{fig:three_regions}.
}
\label{fig:FIFO_vs_3P}
\end{figure}

Finally, as $p_0\uparrow 1$, $T^F\to\infty$ and $D(t)\to\infty$ for each fixed $t>0$, so the FIFO outcome converges to the receiver's no-information benchmark. At $p_0=1$, this limit outcome can be supported in equilibrium if the receiver conjectures no disclosure and, off-path, believes any disclosed good evidence arrived at time $0$. This equilibrium is quite fragile: in addition to requiring this extreme off-path belief (which is also a counterproductive threat), the sender is indifferent over all possible disclosure times. In contrast, in the three-phase equilibrium, the start of transparency stays bounded as $p_0\uparrow 1$ (see Online Appendix).

\subsection{Modeling assumptions}
We assume the state evolves over time. Holding fixed the rest of the model, this assumption is in fact necessary for timestamps to play a nontrivial role; with a constant state, the sender would immediately disclose good evidence and never disclose bad evidence, for all values of the other parameters, with or without timestamps.

We assume that evidence only arrives once. This fits settings where producing additional evidence is infeasible or too costly, and approximates settings with multiple pieces of evidence where one has primary importance to the receiver and the others are secondary. From a technical perspective, this assumption eliminates strategic interactions across disclosures and the effect of the evidence stock on future incentives, allowing us to isolate the role of timestamps. Nonetheless, the good evidence comparison extends in part to multiple arrivals. Without timestamps, the delay under a high prior in Proposition \ref{prop:no_timestamps_necessary_delay} goes through for multiple pieces of evidence; we show this for an absorbing bad state in the Online Appendix. With timestamps, we conjecture that good evidence would continue to be disclosed immediately. We also conjecture that bad evidence would still be disclosed eventually if there are only finitely many pieces of evidence, or if the evidence arrival rate decreases sufficiently with each piece. Fully solving such a model would be more technically challenging as the receiver would have to maintain beliefs about a higher-dimensional evidence status of the sender.

We assume that the sender has no private information about the state other than the evidence he obtains. An alternative model could allow the sender to observe the state at all times. This would allow for stronger off-path belief threats; for example, the receiver could assume that if good evidence is unexpectedly disclosed, the sender must know that the state is \textit{currently} bad. However, since the sender's preferences are state-independent, all equilibrium outcomes in our model would also be supported in equilibria of that  model.

Overall, our assumptions yield a parsimonious model of dynamic disclosure of evidence about an evolving state with and without timestamps. Our results show that removing timestamps makes inferences about evidence age endogenous and fundamentally changes how good and bad evidence are disclosed.

\appendix
\section{Proofs}

\subsection{Proof of Theorem \ref{thm:timestamp_eq}}
The proof proceeds in the following main steps. First, we show that in any equilibrium, good evidence is disclosed immediately (Lemma \ref{lem:timestamps_good_immediate}). Second, we show that bad evidence with any timestamp $\tau$ is disclosed at some (possibly infinite) time $D(\tau)$ increasing in $\tau$ (Corollary \ref{cor:determ_delays}). Third, we characterize $D(\tau)$ uniquely.

For what follows, recall the definitions of $q^{G,\tau}_t$ and $q^{B,\tau}_t$, which are continuous in $t$ and $\tau$ and have  closed-form expressions given by \eqref{qG and qB}. For a fixed $t$, $q^{G,\tau}_t$ is strictly increasing in $\tau$, while $q^{B,\tau}_t$ is strictly decreasing in $\tau$. Fix an arbitrary conjecture about the sender's strategy, and let $p$ be the belief path conditional on no disclosure.

\subsubsection{Immediate disclosure of good evidence with timestamps}

\begin{lemma}\label{lem:timestamps_good_immediate}
    In any equilibrium, if the sender possesses good evidence at any time, he discloses it immediately.
\end{lemma}

In the rest of this section, we prove Lemma \ref{lem:timestamps_good_immediate}. 

\begin{lemma}
If $p_t < q^{G,0}_t$ for all times, then immediate disclosure of good evidence is optimal regardless of its timestamp.
\end{lemma}
\begin{proof}
Suppose the sender has good evidence with timestamp $\tau$ and hasn't disclosed it before time $t$. If he discloses at time $t'\geq t$, his flow payoff is $p_s$ for $s\in[t,t')$ and $q^{G,\tau}_s\geq q^{G,0}_s>p_s$ for all $s\in [t',\infty)$. Immediate disclosure maximizes the flow payoff at all times. 
\end{proof}
    
\begin{lemma}\label{lem:delay_bad}
If the sender possesses bad evidence at any time $t$ and $p_t>q^{B,0}_t$, then there exists $\varepsilon>0$ such that  the sender does not disclose it before time $t+\varepsilon$.
\end{lemma}
\begin{proof}
Suppose the sender has bad evidence with timestamp $\tau$ and hasn't disclosed it before time $t$. Since $p$ and $q^B$ are right-continuous, there exists $\varepsilon>0$ such that $p_u>q^{B,0}_u\geq q^{B,\tau}_u$ for all $u\in [t,t+\varepsilon)$. 
Hence, for any $s\in [t,t+\varepsilon)$, by disclosing at time $t+\varepsilon$ instead of at $s$, the flow payoff is strictly higher on $[s,t+\varepsilon)$ and is unchanged on $[t,s)$ and $[t+\varepsilon,\infty)$.
\end{proof}

Define $t^G=\inf\{t\geq 0: p_t\geq q^{G,0}_t\}$ and $t^B=\inf\{t\geq 0: p_t\leq q^{B,0}_t\}$. By right-continuity of $p$, $q^{G,0}$, and $q^{B,0}$, if $p_0\in (0,1)$, then $t^G,t^B>0$. Right-continuity also implies $p_{t^G}\geq q^{G,0}_{t^G}$ if $t^G$ is finite, and likewise $p_{t^B}\leq q^{B,0}_{t^B}$ if $t^B$ is finite.

We now show that $t^G=\infty$.
Clearly, since $q^{G,0}_t\neq q^{B,0}_t$ for all $t$, we cannot have $t^G=t^B<\infty$. Hence, if $t^G<\infty$, either $t^G<t^B$ or $t^B<t^G$. Toward that goal, the following lemma establishes that if $t^G$ is finite, $p$ must jump above $q^{G,0}$ from below $q^{B,0}$ at time $t^G$.
\begin{lemma}\label{lem:timestamps_p jump from below}
If $t^G<\infty$, then $t^B<t^G$ and $\sup\{t<t^G:p_t\leq q^{B,0}_t\}=t^G$.
\end{lemma}
\begin{proof}

Toward a contradiction, suppose there exists $\varepsilon>0$ such that for all $t\in[t^G-\varepsilon,t^G)$, $p_t>q^{B,0}_t$. By right-continuity, we can choose this $\varepsilon$ such that $p_t>q^{B,0}_t$ for all $t\in[t^G-\varepsilon,t^G+\varepsilon)$. To simplify notation, set $t_0:=t^G-\varepsilon$ and $t_1:=t^G+\varepsilon$. Then by Lemma \ref{lem:delay_bad}, there is no disclosure of bad evidence on $[t_0,t_1)$.
Hence, on $[t_0,t_1)$, $p$ is bounded above by the solution to $\dot{y}_t=\lambda_0(1-y_t)-\lambda_1 y_t$ with initial condition $y_{t_0}=p_{t_0}$. That is, $y_t=\mathbb{P}_{t_0}(\theta_t=1)$ is the  unconditional belief based on Markov switching, starting from $p_{t_0}$. To prove this, let $\mathbb{E}_t$ denote the receiver's expectation operator at time $t$, and $\mathbb{P}_t$ the probability measure at time $t$. By the law of iterated expectations, for all $t\in [t_0,t_1)$, 
\begin{align*}
    y_t=\mathbb{P}_{t_0}(\theta_t=1)=\mathbb{P}_{t_0}(\sigma\leq t)\mathbb{P}_{t_0}(\theta_t=1|\sigma\leq t)+\mathbb{P}_{t_0}(\sigma>t)p_t.
\end{align*}
Because $p_t>q^{B,0}_t$, no bad evidence is disclosed in $[t_0,t_1)$, so any disclosure is of good evidence. The lowest post-disclosure belief from good evidence is $q^{G,0}_t$; good evidence with any later timestamp yields a strictly higher belief. So $\mathbb{P}_{t_0}(\theta_t=1|\sigma\leq t)\geq q^{G,0}_t$.  Moreover, $y_t<q^{G,0}_t$. This follows from the comparison theorem, because $y$ and $q^G$ satisfy the same linear ODE on $[t_0,t_1)$ and differ only in their initial values,  $q^{G,0}_{t_0}>y_{t_0}$. Putting these facts together, 
\begin{align*}
y_t=\mathbb{P}_{t_0}(\sigma\leq t)\mathbb{P}_{t_0}(\theta_t=1|\sigma\leq t)+\mathbb{P}_{t_0}(\sigma>t)p_t&\geq \mathbb{P}_{t_0}(\sigma\leq t)q^{G,0}_t+\mathbb{P}_{t_0}(\sigma>t)p_t\\
&\geq \mathbb{P}_{t_0}(\sigma\leq t)y_t+\mathbb{P}_{t_0}(\sigma>t)p_t.
\end{align*}
Since $\mathbb{P}_{t_0}(\sigma>t)>0$, the last inequality implies that $p_t\leq y_t<q^{G,0}_t$ for all $t\in [t_0,t_1)$, contradicting the definition of $t^G$.
\end{proof}

\begin{lemma}
It must be that $t^G=\infty$.
\end{lemma}
\begin{proof}
If $t^G<\infty$, then by Lemma \ref{lem:timestamps_p jump from below}, at time $t^G$, $p$ must jump from weakly below $q^{B,0}$ to weakly above $q^{G,0}$. Such a discontinuity requires that good or bad evidence is conjectured to be disclosed with an atom of probability at time $t^G$. Now if only an atom of good evidence disclosure is conjectured, $p$ would exhibit a downward discontinuity, since any good evidence disclosure would raise the belief given that $p_{t^G-}\leq q^{B,0}_{t^G}<q^{G,0}_{t^G}$, where $t^G-$ denotes the left limit at $t^G$ and the left limit exists because $p$ is c\`adl\`ag. Thus, at time $t^G$, the receiver must conjecture that bad evidence is disclosed with an atom of probability. But then disclosing bad evidence at time $t^G$ would not be optimal, as $p_{t^G}\ge q^{G,0}_{t^G}>q^{B,0}_{t^G}\geq q^{B,\tau}_{t^G}$ for all $\tau \in [0,t^G]$, so there exists $\varepsilon>0$ such that for all $t\in [t^G,t^G+\varepsilon)$ and all $\tau \in [0,t^G]$, $p_{t}>q^{B,\tau}_{t}$. Therefore, the sender would strictly prefer to wait and disclose at time $t^G+\varepsilon$. 
\end{proof}

Since $t^G=\infty$, we have $p_t<q^{G,0}_t$ for all $t\geq 0$, and thus it is optimal to disclose good evidence immediately, concluding the proof of Lemma \ref{lem:timestamps_good_immediate}.

\subsubsection{Delayed disclosure of bad evidence with timestamps}
If $\lambda_0=0$, it is strictly optimal to never disclose bad evidence. For the rest of the proof, assume $\lambda_0>0$.
\begin{lemma}
In equilibrium, $p$ is continuous.\label{lem:p_continuous}
\end{lemma}
\begin{proof}
If $p$ has a discontinuity at some time $t_0$, then it must be that at time $t_0$, evidence is disclosed with an atom of probability. Since good evidence is disclosed immediately by Lemma \ref{lem:timestamps_good_immediate} and arrives according to an atomless distribution, it must be that bad evidence is disclosed with an atom of probability at time $t_0$. Note that for any type $\tau_0$ with bad evidence that is willing to disclose at $t_0$, it must be that $q^{B,\tau_0}_{t_0}\geq p_{t_0}$, otherwise by right-continuity of $p$ and continuity of $q^{B,\tau_0}$ we would have $p_{t}>q^{B,\tau_0}_{t}$ for all $t\in [t_0,t_0+\varepsilon)$ for some $\varepsilon>0$, and the sender would strictly prefer to disclose at $t_0+\varepsilon$.
    
We now argue that if there is a discontinuity at $t_0$, it must be that $p_{t_0}<p_{t_0-}$. For any $t<t_0$, let $N$ and $G$ denote the events that there is no disclosure or there is good evidence disclosure in $[t,t_0]$ (suppressing dependence on $t$ and $t_0$). Let $B^{t_0}$ denote the event that there is bad evidence disclosure at time $t_0$, and $B^{[t,t_0)}$ the event that there is bad evidence disclosure in $[t,t_0)$. Let $\mathbb{E}_{t}$ and $\mathbb{P}_{t}$ denote expectations and probabilities conditional on no disclosure prior to $t$. By the law of iterated expectations,
\begin{align}\label{eq:law of iterated}
\mathbb{E}_{t}[\theta_{t_0}]=\mathbb{P}_{t}(N)p_{t_0}+\mathbb{P}_{t}(G)\mathbb{E}_{t}[\theta_{t_0}|G]+\mathbb{P}_{t}(B^{[t,t_0)})\mathbb{E}_{t}[\theta_{t_0}|B^{[t,t_0)}]+\mathbb{P}_{t}(B^{t_0}){\mathbb{E}_{t}[\theta_{t_0}|B^{t_0}]}.
\end{align}
As $t\uparrow t_0$, $\mathbb{P}_{t}(G)\to 0$ and $\mathbb{P}_{t}(B^{[t,t_0)})\to 0$, so $\mathbb{P}_{t}(N)+\mathbb{P}_{t}(B^{t_0})\to 1$. To simplify notation, define $\alpha:=\lim_{t\uparrow t_0}\mathbb{P}_{t}(N)$ and  $\beta:=\lim_{t\uparrow t_0} \mathbb{P}_{t}(B^{t_0})=1-\alpha$. We have $\alpha>0$ because the probability of no disclosure in $[t,{t_0})$ is at least the probability that evidence arrives after time $t_0$, which is strictly positive; moreover, $\beta>0$ because there is an atom of bad evidence disclosure at time $t_0$. Solving for $\mathbb{E}_t[\theta_{t_0}|B^{t_0}]$ in \eqref{eq:law of iterated} and taking $t\uparrow t_0$ yields  $\gamma:=\lim_{t\uparrow t_0}\mathbb{E}_{t}[\theta_{t_0}|B^{t_0}]=\frac{p_{t_0-}-\alpha p_{t_0}}{1-\alpha}$. 

Since $\mathbb{E}_{t}[\theta_{t_0}|B^{t_0}]\geq p_{t_0}$ for all $t$, we have $\gamma \geq p_{t_0}$, and therefore $p_{t_0-}\geq p_{t_0}$. A discontinuity at $t_0$ would therefore imply $p_{t_0-}>p_{t_0}$, which implies $\gamma>p_{t_0-}$, so that $\mathbb{E}_{t}[\theta_{t_0}|B^{t_0}]>p_{t_0-}$ for $t$ sufficiently close to $t_0$. As $\mathbb{E}_{t}[\theta_{t_0}|B^{t_0}]$ is a convex combination over $\tau$ of $q^{B,\tau}_{t_0}$ values, there exists some timestamp $\tau<t_0$ for bad evidence for which the sender is willing to wait until $t_0$ to disclose and  $q^{B,\tau}_{t_0}>p_{t_0-}$. By definition of $p_{t_0-}$, for any $\delta>0$, there exists $\varepsilon\in (0,t_0-\tau)$ such that if $t>t_0-\varepsilon$, $|p_{t}-p_{t_0-}|<\delta$. Choosing $\delta$ sufficiently small and $\varepsilon$ accordingly, by continuity of $q^{B,\tau}$, we have $p_t<q^{B,\tau}_{t}$ for all $t\in (t_0-\varepsilon,t_0)$. But then the sender would strictly prefer to disclose at any such $t$ rather than wait until time $t_0$, contradicting optimality. 
\end{proof}

\begin{lemma}[Single-crossing]\label{lem:single_crossing}
If $p_{t_0} \leq q^{B,\tau_0}_{t_0}$ for some $0\leq\tau_0\leq t_0$, $p_{t}<q^{B,\tau_0}_t$ for all $t>t_0$. 
\end{lemma}
\begin{proof}
Assume $p_{t_0} \leq q^{B,\tau_0}_{t_0}$ for some $0\leq \tau_0\leq t_0$. Define $t^*=\inf\{t>t_0:p_{t}\geq q^{B,\tau_0}_{t}\}$. We aim to show that $t^*=\infty$. First, we rule out $t^*\in (t_0,\infty)$, and then we rule out $t^*=t_0$. Toward a contradiction, suppose $t^*\in (t_0,\infty)$. 

For each $\tau\leq t^*$, define $F(\tau)=\max_{t\in[\max\{t_0,\tau\},t^*]} \left(q^{B,\tau}_t-p_{t}\right)$. Since $q^B$ is continuous in both arguments and $p$ is continuous by Lemma \ref{lem:p_continuous}, $F$ is continuous by the maximum theorem. Moreover, $F$ is strictly decreasing since $q^B$ is strictly decreasing in $\tau$ and the set $[\max\{t_0,\tau\},t^*]$ is weakly decreasing in $\tau$ in the set inclusion sense;  it satisfies $F(\tau_0)>0$ since $t^*>t_0$ by assumption, and by the definition of $t^*$ we cannot have $p_t\geq q^{B,\tau_0}_t$ for any $t\in (t_0,t^*)$; and it satisfies $F(t^*)<0$ since $q^{B,t^*}_{t^*}=0<p_{t^*}$, where the latter inequality holds because $\phi_{t^*}>0$ and there is a positive probability of no evidence arrival. Thus, there exists $\overline\tau>\tau_0$ such that $F(\tau)\geq 0$ if and only if $\tau\leq \overline \tau$. Since $F(\overline\tau)=0$, there exists some $t_1\in [\max\{t_0,\overline\tau\},t^*)$ where $q^{B,\overline\tau}_{t_1}=p_{t_1}$, with $q^{B,\overline\tau}_{t}\leq p_{t}$ for all $t\in [\max\{t_0,\overline\tau\},t^*]$.

With a slight abuse of notation, let $\phi_a(x)$ denote the posterior belief based on hidden Markov switching over $a$ units of time, starting from belief $x$.
Thus, for $t\geq t_1$, $\phi_{t-t_1}(p_{t_1})$ is the time-$t$ value
of the solution to $\dot{y}_t=\lambda_0(1-y_t)-\lambda_1y_t$
subject to $y_{t_1}=p_{t_1}$. In words,
$\phi_{t-t_1}(p_{t_1})$ is the time-$t$ belief when the belief at time
$t_1$ is $p_{t_1}$. Note that $q_t^{B,\tau_0}$ solves the same ODE. Since $p_{t_1}<q^{B,\tau_0}_{t_1}$, it follows that $\phi_{t^*-t_1}(p_{t_1})<q^{B,\tau_0}_{t^*}=p_{t^*}$. Now let $N$, $G$, and $B$ denote the events that no disclosure, good disclosure, and bad disclosure, respectively, occurs in $[t_1,t^*]$. By the law of iterated expectations,
\begin{align*}
\mathbb{E}_{t_1}[\theta_{t^*}]=\phi_{t^*-t_1}(p_{t_1})=\mathbb{P}_{t_1}(N)p_{t^*}+\mathbb{P}_{t_1}(G)\mathbb{E}_{t_1}[\theta_{t^*}|G]+\mathbb{P}_{t_1}(B)\mathbb{E}_{t_1}[\theta_{t^*}|B].
\end{align*}
Now $\mathbb{E}_{t_1}[\theta_{t^*}|G]>q^{B,\tau_0}_{t^*}=p_{t^*}$. Hence,
\begin{align*}
\phi_{t^*-t_1}(p_{t_1})\geq (1-\mathbb{P}_{t_1}(B))p_{t^*}+\mathbb{P}_{t_1}(B)\mathbb{E}_{t_1}[\theta_{t^*}|B].
\end{align*}

Recall that $p_{t^*}>\phi_{t^*-t_1}(p_{t_1})$, and $\mathbb{P}_{t_1}(B)<1$, so we must have
\begin{align*}
\phi_{t^*-t_1}(p_{t_1})> (1-\mathbb{P}_{t_1}(B))\phi_{t^*-t_1}(p_{t_1})+\mathbb{P}_{t_1}(B)\mathbb{E}_{t_1}[\theta_{t^*}|B]
\implies 
\phi_{t^*-t_1}(p_{t_1})>\mathbb{E}_{t_1}[\theta_{t^*}|B],
\end{align*}
where $\mathbb{P}_{t_1}(B)>0$ as the first inequality cannot hold if $\mathbb{P}_{t_1}(B)=0$. 
Since $\mathbb{E}_{t_1}[\theta_{t^*}|B]=\mathbb{E}_{t_1}[q^{B,\tau}_{t^*}|B]$, this means there must be some timestamp $\tau$ and some time $t\in[t_1,t^*]$ for which bad evidence is disclosed at time $t$ and $q^{B,\tau}_{t^*}<\phi_{t^*-t_1}(p_{t_1})$. Since $\phi_{u-t_1}(p_{t_1})=q^{B,\overline\tau}_u$ for all $u\in [t_1,t^*]$ (as  $q^{B,\overline\tau}_\cdot$ solves the same ODE with same initial value $q^{B,\overline\tau}_{t_1}=p_{t_1}=\phi_{0}(p_{t_1})$), we have $q^{B,\tau}_{t^*}<q^{B,\overline\tau}_{t^*}$, and because $q^{B,\cdot}_s$ is strictly decreasing in the timestamp on $[0,s]$ for each fixed $s$, this implies $\tau>\overline\tau$. 
But then we have $F(\tau)<0$, so $p_t>q^{B,\tau}_t$ for all $t\in (\max\{t_0,\tau\},t^*]$. By continuity, this inequality also holds on $(t^*,t^*+\varepsilon)$ for some small $\varepsilon>0$. This implies that the sender would strictly prefer to delay disclosure until time $t^*+\varepsilon$, a contradiction. 

Next, we rule out the case $t^*=t_0$ by adapting arguments from the previous case. Toward a contradiction, suppose that $t^*=t_0$. Then for all $\varepsilon>0$ there exists $\hat t\in (t_0,t_0+\varepsilon)$ such that $p_{\hat t}\geq q^{B,\tau_0}_{\hat t}$. By arguments similar to those used in the previous case, there exists some timestamp $\overline\tau\geq \tau_0$ and time $t_1\in [\max\{\overline\tau,t_0\},\hat t)$ such that $p_{t_1}=q^{B,\overline\tau}_{t_1}$ and $p_t\geq q^{B,\overline\tau}_t$ for all $t\in [\max\{\overline\tau,t_0\},\hat t]$. Since $\phi_{u-t_1}(p_{t_1})$ and $q^{B,\overline\tau}_u$ solve the same ODE and coincide at $t_1$, we have $\phi_{u-t_1}(p_{t_1})=q^{B,\overline\tau}_u$ for all $u\geq t_1$, and in particular for $u=\hat t$.

Let $N$, $G$, and $B$ be the events of no disclosure, good evidence disclosure, and bad evidence disclosure on $[t_1,\hat t]$. By similar logic to that used in the proof for  the $t^*\in (t_0,\infty)$ case, we have $\phi_{\hat{t}-t_1}(p_{t_1})=q^{B,\overline \tau}_{\hat t}\leq q^{B,\tau_0}_{\hat t}\leq p_{\hat t}$ and also $\mathbb{E}_{t_1}[\theta_{\hat t}|G]>q^{B,\tau_0}_{\hat t}\geq q^{B,\overline \tau}_{\hat t}$. Thus, by the law of iterated expectations, 
\begin{align*}
\mathbb{E}_{t_1}[\theta_{\hat t}]=q^{B,\overline \tau}_{\hat t} &=\mathbb{P}_{t_1}(N)p_{\hat t}+\mathbb{P}_{t_1}(G)\mathbb{E}_{t_1}[\theta_{\hat t}|G]+\mathbb{P}_{t_1}(B)\mathbb{E}_{t_1}[\theta_{\hat t}|B]\\
&> \mathbb{P}_{t_1}(N)q^{B,\overline \tau}_{\hat t}+\mathbb{P}_{t_1}(G)q^{B,\overline \tau}_{\hat t}+\mathbb{P}_{t_1}(B)\mathbb{E}_{t_1}[\theta_{\hat t}|B].
\end{align*}
As in the previous case, the strict inequality implies $\mathbb{P}_{t_1}(B)>0$. Thus, $
q^{B,\overline\tau}_{\hat t}>\mathbb{E}_{t_1}[\theta_{\hat t}|B].$
Since the expectation is a convex combination of $q^{B,\tau}_{\hat t}$ values over $\tau\leq \hat t$, there must be disclosure at some time $t\in[t_1,\hat t]$ of bad evidence with some timestamp $\tau>\overline \tau$. But at $t$, we have $p_t\geq q^{B,\overline\tau}_t>q^{B,\tau}_t$. By continuity in $t$, this inequality holds for all times in an interval $[t,t+\varepsilon)$ for some $\varepsilon>0$, and the sender would strictly prefer to delay disclosure until $t+\varepsilon$, contradicting optimality. This rules out $t^*=t_0$.
\end{proof}

\begin{cor}[Deterministic delays]\label{cor:determ_delays}
In equilibrium, bad evidence with timestamp $\tau$ is disclosed at some time $D(\tau)\in (\tau,\infty]$. Moreover, there exists $\overline\tau\in [0,\infty]$ such that for $\tau \in [0,\overline\tau)$, $D(\tau)$ is strictly increasing and finite. If $\overline\tau<\infty$, then  $D(\tau)=\infty$ for all $\tau \geq \overline\tau$.
\end{cor}
\begin{proof}
Suppose that at some time $t_0$, the sender has bad evidence with timestamp $\tau_0\leq t_0$. Let $D(\tau_0,t_0):=\inf\{t>t_0:p_t\leq q^{B,\tau_0}_t\}$. We claim that disclosing at time $D(\tau_0,t_0)$ is optimal. If $D(\tau_0,t_0)=\infty$, then we have $p_t>q^{B,\tau_0}_t$ for all $t\geq t_0$, and clearly it is optimal for the sender to never disclose the evidence. If $D(\tau_0,t_0)<\infty$, then we have (i) $p_t>q^{B,\tau_0}_t$ for all $t<D(\tau_0,t_0)$, which implies disclosing at time $D(\tau_0,t_0)$ yields a higher payoff than disclosing at any earlier time, and (ii) $p_t<q^{B,\tau_0}_t$ for all $t>D(\tau_0,t_0)$ by Lemma \ref{lem:single_crossing}, which implies that disclosing at time $D(\tau_0,t_0)$ yields a higher payoff than disclosing at any later time. On path, it is therefore optimal to disclose at time $D(\tau):=D(\tau,\tau)$. Note that since $q^{B,\tau}_\tau=0<p_\tau$, right-continuity gives $D(\tau)>\tau$. 

Define $\overline\tau:=\inf \{\tau\geq 0: D(\tau)=\infty\}$.  Since $q^{B,\tau}_t$ is strictly increasing in $t$ and strictly decreasing in $\tau$, we have that $D(\tau)$ is increasing on $[0,\infty)$ and strictly increasing on $[0,\overline\tau)$. Finally, we have $D(\overline\tau)=\infty$, because otherwise $D(\overline\tau)<\infty$, which means $q^{B,\overline\tau}_t$ crosses $p_t$ at some finite time, and stays strictly above it after $t$ by Lemma \ref{lem:single_crossing}; hence, by continuity, $q^{B,\tau}_t$ crosses $p_t$ for all $\tau>\overline\tau$ sufficiently close to $\overline\tau$, contradicting the definition of $\overline\tau$. 
\end{proof}

We now turn to characterizing $\tau\mapsto D(\tau)$. Note that, in equilibrium, when $p_t=q^{B,s}_t$ for some $s$, then at this $t$, $p_t=p_{t,s}^{\dagger}$, where $p_{t,s}^{\dagger}$ is defined in \eqref{eq:p_dagger}.
Hence, we solve \eqref{eq:aux_belief_crossing}.

\begin{lemma}
\label{lem: belief intersect}If $\lambda_0,\lambda_1>0$, then for all $s\ge0$, there is a unique
$t^{\dagger}(s)\in (s,\infty)$ solving \eqref{eq:aux_belief_crossing}, and it is continuously differentiable and strictly increasing in $s$.
\end{lemma}
\begin{proof}
Expanding \eqref{eq:aux_belief_crossing} and multiplying through by the denominator yields
\begin{align}
e^{-\mu t}\phi_{t}+\int_{s}^{t}\mu e^{-\mu\tau}(1-\phi_{\tau})q^{B,\tau}_t d\tau&=q^{B,s}_t\left[e^{-\mu t}+\int_{s}^{t}\mu e^{-\mu\tau}(1-\phi_{\tau})d\tau\right]\notag\\
\iff\phi_t-q^{B,s}_t&=\int_{s}^{t}\mu e^{-\mu(\tau-t)} (1-\phi_{\tau})(q^{B,s}_t-q^{B,\tau}_t)d\tau.\label{eq:belief equal IVT}
\end{align}

We establish existence by the intermediate value theorem. At $t=s$, the left-hand side is $\phi_s-0>0$, while the right-hand side vanishes. As $t\to\infty$, the left-hand side tends to $p^*-p^*=0$. Meanwhile, for fixed $\varepsilon>0$, the right-hand side is bounded below by
\begin{align}
    \int_{t-\varepsilon}^t\mu e^{-\mu(\tau-t)} (1-\phi_{\tau})(q^{B,s}_t-q^{B,\tau}_t)d\tau\geq \int_{t-\varepsilon}^t\mu e^{-\mu\varepsilon} (1-\phi_{\tau})(q^{B,s}_t-q^{B,t-\varepsilon}_t)d\tau.\label{eq:belief equal IVT 2}
\end{align}
As $t\to\infty$, we have $\phi_t\to p^*<1$ given that $\lambda_1>0$, and we also have $q^{B,s}_t\to p^*>0$ under the assumption that $\lambda_0>0$. In addition, $q^{B,t-\varepsilon}_t<p^*(\lambda_0+\lambda_1)\varepsilon$. Hence, for small $\varepsilon$ and large $t$ the integrand on the right-hand side of \eqref{eq:belief equal IVT 2} is positive and uniformly bounded away from zero, so the right-hand side of \eqref{eq:belief equal IVT} is also bounded away from zero for large $t$. Since both the left- and right-hand side of \eqref{eq:belief equal IVT} are continuous in $t$, there exists a solution $t^{\dagger}(s)\in (s,\infty)$.

Uniqueness follows from showing that $p^{\dagger}_{t,s}$ crosses $q^{B,s}_t$ from above at $t^{\dagger}(s)$.  Given Markov switching, $\dot{q}^{B,s}_t = \lambda_0 (1 - q^{B,s}_t) - \lambda_1 q^{B,s}_t$. For fixed $s$, we can write the evolution of $p^{\dagger}_{t,s}$ as
\begin{align}\label{pdagger derivative}
    \dot{p}^{\dagger}_{t,s} = \lambda_0 (1 - p^{\dagger}_{t,s}) - \lambda_1 p^{\dagger}_{t,s} - \frac{\mu e^{-\mu t} \phi_t}{e^{-\mu t} + \int_s^t \mu e^{-\mu \tau} (1 - \phi_\tau)\, d\tau}\, (1 - p^{\dagger}_{t,s}).
\end{align}
At $t^\dagger(s)$, $p^{\dagger}_{t^\dagger(s),s} = q^{B,s}_{t^\dagger(s)}$, so
$$ \dot{q}^{B,s}_{t^\dagger(s)} -\dot{p}^{\dagger}_{t^\dagger(s),s} = \frac{\mu e^{-\mu t^\dagger(s)} \phi_{t^\dagger(s)}}{e^{-\mu t^\dagger(s)} + \int_s^{t^\dagger(s)} \mu e^{-\mu \tau} (1 - \phi_\tau)\, d\tau}\, (1 - p^{\dagger}_{t^\dagger(s),s}) > 0.$$

Now the crossing condition \eqref{eq:aux_belief_crossing} can be written in the form $K(t^\dagger(s),s)=0$, where $K(t,s):=q^{B,s}_t-p^\dagger_{t,s}$. In addition, we have just shown that $K_t(t^\dagger(s),s)>0$ for all $s$. Hence, by the implicit function theorem, $t^\dagger(s)$ is continuously differentiable, and $t^{\dagger\prime}(s)=-\frac{K_s(t^\dagger(s),s)}{K_t(t^\dagger(s),s)}$. We have $K_s(t^\dagger(s),s)=\frac{\partial q^{B,s}_{t^\dagger(s)}}{\partial s}-\frac{\partial p^\dagger_{t^\dagger(s),s}}{\partial s}$. We claim that $\frac{\partial p^\dagger_{t^\dagger(s),s}}{\partial s}=0$. Directly calculating this derivative yields 
\begin{align*}
   \frac{\partial p^\dagger_{t^\dagger(s),s}}{\partial s}=(p^\dagger_{t,s}-q^{B,s}_t) \frac{\mu e^{-\mu s}(1-\phi_s)}{e^{-\mu t}+\int_{s}^{t}\mu e^{-\mu\tau}(1-\phi_{\tau})d\tau}|_{t=t^\dagger(s)}=0
\end{align*}
where we have used \eqref{eq:aux_belief_crossing} at $t^\dagger(s)$. Thus $K_s(t^\dagger(s),s)=\frac{\partial q^{B,s}_{t^\dagger(s)}}{\partial s}<0$, and $t^{\dagger\prime}(s)>0$.
\end{proof}

It follows that for each timestamp $s$ for which bad evidence is eventually disclosed, i.e. for each $s\in [0,\bar\tau)$ in Corollary \ref{cor:determ_delays}, the optimal time to disclose is $D(s)=t^\dagger(s)$. We now argue that (under $\lambda_0,\lambda_1>0$), $\bar\tau=\infty$. First, we cannot have $\bar\tau=0$, because in this case $p_t=p^{\dagger}_{t,0}$ at all times, but we have already shown that $t^\dagger(0)<\infty$, so disclosure at time $D(0)=t^\dagger(0)$ is optimal. Next, we rule out $\bar\tau\in (0,\infty)$. By way of contradiction, suppose $\bar\tau\in (0,\infty)$. Since $t^\dagger$ is continuous, $\lim_{s\uparrow\bar\tau}D(s)=\lim_{s\uparrow\bar\tau}t^\dagger(s)=T:=t^\dagger(\bar\tau)<\infty$, while $D(\bar\tau)=\infty$. This means that by time $T$, bad evidence has been disclosed if and only if it arrived strictly before $\bar\tau$, and since evidence arrival is atomless $p_T=p^\dagger_{T,\bar\tau}=q^{B,\bar\tau}_{T}$, and thus it is optimal to disclose bad evidence from $\bar\tau$ at time $T$, a contradiction. Given that $\bar\tau=\infty$, $t^\dagger(s)$ fully characterizes the sender's disclosure of bad evidence for all $s\in [0,\infty)$. Finally, since $t^\dagger(s)$ is strictly increasing and continuous with range $[t^\dagger(0),\infty)$, for each $t$ in this interval, we have $t=t^\dagger(s)$ for some $s\geq 0$, and therefore $p_t=p^\dagger_{t^\dagger(s),s}$. For $t\in [0,t^\dagger(0))$, we have $p_t=p^\dagger_{t,0}$. The strategy can be easily recovered in CDF form: for all $\tau\leq t<\infty$ and $s\geq t$, $H(s;G,\tau,t)=1$, and $H(s;B,\tau,t)=\mathbbm{1}_{\{s\geq t^{\dagger}(\tau)\}}$.

\subsection{Proof of Proposition \ref{prop:timestamp properties}}
By definition, $t^{\dagger}(s) = s + a^{*}(s)$. Part (i) is established by Lemma \ref{lem: belief intersect}.

For part (ii), we sign $a^{*\prime}(s)$. Using calculations from the proof of Lemma \ref{lem: belief intersect}, 
\begin{align*}
 a^{*\prime}(s)=   t^{\dagger\prime}(s)-1=-\frac{K_s(t^\dagger(s),s)}{K_t(t^\dagger(s),s)}-1&=\frac{-\partial q^{B,s}_{t^\dagger(s)}/\partial s-K_t(t^\dagger(s),s)}{K_t(t^\dagger(s),s)}\\
   &=\frac{\partial q^{B,s}_{t^\dagger(s)}/\partial t-K_t(t^\dagger(s),s)}{K_t(t^\dagger(s),s)}=\frac{\partial p^{\dagger}_{t^\dagger(s),s}/\partial t}{K_t(t^\dagger(s),s)}.
\end{align*}
Thus, $a^{*\prime}(s)$ has the same sign as $\partial p^{\dagger}_{t^\dagger(s),s}/\partial t$.
Now fix $a=t^\dagger(s)-s$. Then $$\frac{dp_{s+a,s}^{\dagger}}{ds}=\frac{\partial p_{t,s}^{\dagger}}{\partial t}\bigg|_{t=s+a}+\frac{\partial p_{t,s}^{\dagger}}{\partial s}\bigg|_{t=s+a}\implies \frac{dp_{s+a,s}^{\dagger}}{ds}\bigg|_{a=a^{*}(s)}=\frac{\partial p_{t,s}^{\dagger}}{\partial t}\bigg|_{t=s+a^{*}(s)},$$
where we have used that $\frac{\partial p_{t,s}^{\dagger}}{\partial s}|_{t=s+a^{*}(s)}=0$. 
Then, it suffices to derive necessary and sufficient conditions under which $dp_{s+a,s}^{\dagger}/ds|_{a=a^{*}(s)}$ is (strictly) positive for any $s\ge0$. By definition,
$$p_{s+a,s}^{\dagger}=\frac{e^{-\mu(s+a)}\phi_{s+a}+\int_{s}^{s+a}\mu e^{-\mu\tau}(1-\phi_{\tau})q_{s+a}^{B,\tau}d\tau}{e^{-\mu(s+a)}+\int_{s}^{s+a}\mu e^{-\mu\tau}(1-\phi_{\tau})d\tau}.$$
By a change of variables, let $\tau=s+v$ where $v\in[0,a]$. Then $d\tau=dv$ (since $s$ is fixed). After some simplifying, we get
\begin{align}
p_{s+a,s}^{\dagger}=\frac{e^{-\mu a}\phi_{s+a}+\int_{0}^{a}\mu e^{-\mu v}(1-\phi_{s+v})q_{s+a}^{B,s+v}dv}{e^{-\mu a}+\int_{0}^{a}\mu e^{-\mu v}(1-\phi_{s+v})dv}.\label{eq:pdagger a}
\end{align}
Note that $q_{s+a}^{B,s+v}=q_{a}^{B,v}=p^{*}\left(1-e^{-(\lambda_{0}+\lambda_{1})(a-v)}\right)$ is independent of $s$ because the Markov process is time-homogeneous. This means $p_{s+a,s}^{\dagger}$ changes with $s$ only through its changes in $\phi_{\cdot}$. Differentiate both sides with respect to $s$ and note that $d\phi_{s+x}/ds=d\phi_{t}/dt|_{t=s+x}=:\dot{\phi}_{s+x}$. After simplifying, $$\frac{dp_{s+a,s}^{\dagger}}{ds}=\frac{e^{-\mu a}\dot{\phi}_{s+a}+\int_{0}^{a}\mu e^{-\mu v}\dot{\phi}_{s+v}\left(p_{s+a,s}^{\dagger}-q_{s+a}^{B,s+v}\right)dv}{e^{-\mu a}+\int_{0}^{a}\mu e^{-\mu v}(1-\phi_{s+v})dv}.$$
Evaluate at $a=a^{*}(s)$. Because the denominator is positive, $\frac{dp_{s+a,s}^{\dagger}}{ds}\big|_{a=a^{*}(s)}$ has the same sign as the numerator: $$e^{-\mu a^{*}(s)}\dot{\phi}_{s+a^{*}(s)}+\int_{0}^{a^{*}(s)}\mu e^{-\mu v}\dot{\phi}_{s+v}\left(p_{s+a^{*}(s),s}^{\dagger}-q_{s+a^{*}(s)}^{B,s+v}\right)dv.$$
Because $p_{s+a^{*}(s),s}^{\dagger}=q_{s+a^{*}(s)}^{B,s}$ so $p_{s+a^{*}(s),s}^{\dagger}-q_{s+a^{*}(s)}^{B,s+v}=q_{s+a^{*}(s)}^{B,s}-q_{s+a^{*}(s)}^{B,s+v}\ge0$ for all $v\in[0,a^{*}(s)]$ with equality only at $v=0$. Moreover, $\dot{\phi}_{t}=-(\lambda_{0}+\lambda_{1})(p_{0}-p^{*})e^{-(\lambda_{0}+\lambda_{1})t}$ for all $t$, so $\dot{\phi}_{t}>0$ if and only if $p_{0}<p^{*} $, and $\dot{\phi}_{s+a^{*}(s)}$ and $\dot{\phi}_{s+v}$ have the same sign. Thus, the numerator is positive if and only if $p_{0}<p^{*}$, negative if $p_0>p^*$, and $0$ if $p_0=p^*$. This proves (ii). 

For (iii), define $a^*(s,p_0)=t^\dagger(s,p_0)-s$, making dependence on $p_0$ explicit. By part (ii), $s\mapsto a^*(s,p_0)$ is strictly decreasing if $p_0> p^*$, constant if $p_0=p^*$, and strictly increasing if $p_0<p^*$. For $p_0\geq p^*$, the monotonicity and boundedness from below give the existence of a nonnegative limit denoted by $\alpha^*(p_0)$. Using a change of variables $v=\tau-s$ in \eqref{eq:aux_belief_crossing}, we have
\begin{align*}
    &p^\dagger_{s+a^*(s,p_0),s}=q^{B,s}_{s+a^*(s,p_0)}\\
    \iff &\frac{e^{-\mu a^*(s,p_0)}\phi_{s+a^*(s,p_0)} +\int_0^{a^*(s,p_0)} \mu e^{-\mu v}(1-\phi_{v+s})q^{B,v}_{a^*(s,p_0)}\,dv}{e^{-\mu a^*(s,p_0)}+\int_0^{a^*(s,p_0)} \mu e^{-\mu v}(1-\phi_{v+s})\,dv}=p^*(1-e^{-(\lambda_0+\lambda_1)a^*(s,p_0)}).
\end{align*}

Taking $s\to\infty$ and arguing via dominated convergence yields
\begin{align*}
    \frac{e^{-\mu \alpha^*(p_0)}p^* +\int_0^{\alpha^*(p_0)} \mu e^{-\mu v}(1-p^*)q^{B,v}_{\alpha^*(p_0)}\,dv}{e^{-\mu \alpha^*(p_0)}+\int_0^{\alpha^*(p_0)} \mu e^{-\mu v}(1-p^*)\,dv}=p^*(1-e^{-(\lambda_0+\lambda_1)\alpha^*(p_0)}).
\end{align*}
This equation is independent of $p_0$ and is exactly the equation for $p_0=p^*$ at $s=0$, which has a unique solution $\alpha^*:=\alpha^*(p^*)\in (0,\infty)$. Thus, we have $\alpha^*(p_0)=\alpha^*$ for all $p_0\in [p^*,1)$.

Next, consider $p_0<p^*$, and let $\alpha^*(p_0)$ again denote the (possibly infinite) limit. If the limit is finite, the same argument above shows that it must equal $\alpha^*$, and we are done. Hence, we need only rule out $\alpha^*(p_0)=+\infty$. Choose any finite $K>\alpha^*$. Then by uniqueness of $\alpha^*$ and the direction of crossing, we have $p^{\dagger}_{K,0}<q^{B,0}_K,$ and thus for sufficiently large $s$ (by convergence as $s\to \infty$ for fixed $K$), $p^{\dagger}_{s+K,s}<q^{B,s}_{s+K}$. 
Hence, for sufficiently large $s$, $a^*(s,p_0)\leq K$. This implies that the limit is indeed finite, so it must be $\alpha^*$.

\subsection{Proof of Corollary \ref{cor:no-disclosure-belief}}
By Proposition \ref{prop:timestamp properties}(i), $s+a^*(s)=t^\dagger(s)>s$ is strictly increasing in $s$; moreover, it is continuously differentiable and has range $[a^*(0),\infty)$ by the proof of Lemma \ref{lem: belief intersect}, which implies that for each $t\ge a^*(0)$, there is a unique $\tau(t)\ge 0$ satisfying $\tau(t)+a^*(\tau(t))=t$. Moreover $\tau(t)$ is strictly increasing and continuous in $t$. The equilibrium condition at each $t$ says
\[
p_t=p_{t,\tau(t)}^{\dagger}=q_t^{B,\tau(t)}
=p^*\left(1-e^{-(\lambda_0+\lambda_1)(t-\tau(t))}\right)
=p^*\left(1-e^{-(\lambda_0+\lambda_1)a^*(\tau(t))}\right).
\]
Since the right-hand side is strictly increasing in $a^*(\tau(t))$ and $\tau(t)$ itself is  strictly increasing, $p_t$ inherits the monotonicity of $a^*(s)$. Proposition \ref{prop:timestamp properties} (ii) then implies that $p_t$ is strictly increasing when $p_0<p^*$, strictly decreasing when $p_0>p^*$, and constant when $p_0=p^*$.

Finally, as $t\to\infty$, $\tau(t)\to\infty$, so by Proposition \ref{prop:timestamp properties} (iii), $a^*(\tau(t))\to\alpha^*$ and $p_t\to p^*\left(1-e^{-(\lambda_0+\lambda_1)\alpha^*}\right)
<p^*,$
where the inequality follows because $\alpha^*$ is finite.

\subsection{Proof of Proposition \ref{prop:min_princ}}
For $\lambda_0=0$, define $a^*(0)=+\infty$ (and retain the definition from Theorem \ref{thm:timestamp_eq} for $\lambda_0>0$). For each $t\geq a^*(0)$, define $\tau(t)$ as in the proof of Corollary \ref{cor:no-disclosure-belief}. For $t<a^*(0)$, set $\tau(t)=-\infty$. 

Now fix any $t\geq 0$. Let $\nu^G_s$ and $\nu^B_s$ denote the probabilities in equilibrium that good evidence and bad evidence, respectively, obtained at time $s$ are disclosed by time $t$. Fix an arbitrary alternative disclosure strategy and let $\tilde{\nu}^G_s$ and $\tilde{\nu}^B_s$ denote the analogous probabilities. Note that $\nu^G_s=1$ for all $s\in [0,t]$, and $\nu^B_s=1$ for all $s\leq \tau(t)$ and $\nu^B_s=0$ for all $s>\tau(t)$. Let $p_t$ and $\tilde{p}_t$ denote the beliefs conditional on no disclosure in the equilibrium and under the alternative disclosure strategy, respectively. Finally, let $N_t$ and $\tilde{N}_t$ denote the probabilities of no disclosure under the equilibrium and alternative strategies, respectively. 

By construction, we have (i) $p_t\leq q^{G,s}_t$ for all $s\in [0,t]$,  (ii) $p_t\leq q^{B,s}_t$ for $s\in [0,\tau(t)]$, and (iii) $q^{B,s}_t\leq p_t$ for $s\in(\tau(t),t]$.

By the law of iterated expectations, 
\begin{align}
\phi_t&=N_t p_t + \int_{0}^t\phi_s \mu e^{-\mu s}\nu^G_s q^{G,s}_t\,ds+\int_{0}^t(1-\phi_s) \mu e^{-\mu s}\nu^B_s q^{B,s}_t\,ds\label{eq:min_princ_LIE_1}\\
\text{and}\qquad \phi_t&=\tilde{N}_t \tilde{p}_t + \int_0^t\phi_s \mu e^{-\mu s}\tilde{\nu}^G_s q^{G,s}_t\,ds+\int_0^t(1-\phi_s) \mu e^{-\mu s}\tilde{\nu}^B_s q^{B,s}_t\,ds,\label{eq:min_princ_LIE_2}
\end{align}
where the belief at time $t$ after disclosure (at any time) of good or bad evidence with timestamp $s$, namely $q^{G,s}_t$ or $q^{B,s}_t$, is independent of the conjectured strategy. Equating the right hand sides of \eqref{eq:min_princ_LIE_1} and \eqref{eq:min_princ_LIE_2} and subtracting $p_t$ from both yields
\begin{align}
&0+  \int_0^t\phi_s \mu e^{-\mu s}\nu^G_s (q^{G,s}_t-p_t)\,ds+\int_0^t(1-\phi_s) \mu e^{-\mu s}\nu^B_s (q^{B,s}_t-p_t)\,ds \notag \\  =&\tilde{N}_t (\tilde{p}_t-p_t) + \int_0^t\phi_s \mu e^{-\mu s}\tilde{\nu}^G_s (q^{G,s}_t-p_t)\,ds+\int_0^t(1-\phi_s) \mu e^{-\mu s}\tilde{\nu}^B_s (q^{B,s}_t-p_t)\,ds\notag\\
\begin{split}\iff &\tilde{N}_t (\tilde{p}_t-p_t)=\int_0^t\phi_s \mu e^{-\mu s}(\nu^G_s -\tilde{\nu}^G_s)(q^{G,s}_t-p_t)\,ds\\&\qquad\qquad+\int_0^t(1-\phi_s) \mu e^{-\mu s}(\nu^B_s -\tilde{\nu}^B_s) (q^{B,s}_t-p_t)\,ds.\label{eq:min_princ_diff}\end{split}
    \end{align}
Now for all $s\in [0,t]$, we have $1=\nu^G_s\geq \tilde{\nu}^G_s$ and $q^{G,s}_t-p_t\geq 0$, so the first integral on the right hand side of \eqref{eq:min_princ_diff} is nonnegative. We then turn to the integrand of the second integral. For $s\in [0,\tau(t)]$, we have $1=\nu^B_s\geq \tilde{\nu}^B_s$ and $q^{B,s}_t\geq p_t$, so the integrand is nonnegative. And for $s\in (\tau(t),t]$, we have $0=\nu^B_s\leq \tilde{\nu}^B_s$ and $q^{B,s}_t\leq p_t$, so again the integrand is nonnegative. Hence, the right hand side of \eqref{eq:min_princ_diff} is nonnegative. As
$\tilde{N}_t\geq e^{-\mu t}>0$, we must have $\tilde p_t\geq p_t$.

\subsection{Proof of Proposition \ref{prop:no_timestamps_necessary_delay}}
Define $\Lambda:=\lambda_0+\lambda_1$. By way of contradiction, suppose there exists an equilibrium and some $\delta>0$ such that $S_t=0$ for all $t\in (0,\delta)$, where $S_t$ is defined, as in Section \ref{subsec:no_timestamps_good}, as the probability that the sender has good evidence at time $t$ conditional on no disclosure by $t$. This means that for almost all $\tau\in (0,\delta)$, good evidence arriving at $\tau$ is disclosed immediately, and for almost all dates $t\in (0,\delta)$, good evidence disclosed at $t$ is interpreted as fresh, i.e. $q_t=1$. Now suppose the sender obtains good evidence at time $0$. Let $F(T)$ denote the sender's discounted payoff from disclosing at time $T$. By disclosing immediately, the sender's payoff is $F(0)=\int_0^\infty e^{-rt}\left[p^*+(1-p^*)e^{-\Lambda t}\right]\,dt=\frac{\lambda_0+r}{r(\Lambda+r)}$. By instead disclosing at some time $\varepsilon \in (0,\delta)$ (among those for which $q_\varepsilon=1$), the sender's payoff is $F(\varepsilon)=\int_0^\varepsilon e^{-rt}p_t\,dt+\int_\varepsilon^\infty e^{-rt}\left[p^*+(1-p^*)e^{-\Lambda(t-\varepsilon)}\right]\,dt=\int_0^\varepsilon e^{-rt}p_t\,dt+e^{-r\varepsilon}F(0)$. If $p_0>\frac{\lambda_0+r}{\Lambda+r}$, then since $p_t$ is right-continuous, we have for sufficiently small $\varepsilon>0$ that $$\int_0^\varepsilon e^{-rt}p_t\,dt>\int_0^\varepsilon e^{-rt}\left(\frac{\lambda_0+r}{\Lambda+r}\right)\,dt=\frac{1-e^{-r\varepsilon}}{r}\frac{\lambda_0+r}{\Lambda+r}=(1-e^{-r\varepsilon})F(0).$$ Hence, $F(\varepsilon)>(1-e^{-r\varepsilon})F(0)+e^{-r\varepsilon}F(0)=F(0)$. Thus, delay until $\varepsilon$ is strictly preferable to immediate disclosure. By continuity, delay until $\varepsilon$ is also strictly preferable to immediate disclosure for all types in $[0,\eta)$ for some sufficiently small $\eta\in (0,\varepsilon)$, a contradiction.

\subsection{Proof of Lemma \ref{lem:no_good_bang_bang}}
Let $S$, $G_{0}$, and $G_1$ be defined as in Section \ref{subsec:construction}. 
For all finite $\tau\geq 0$ and all $T\in [\tau,\infty]$, define $U(T;\tau)$ as the payoff from time $\tau$ of disclosing good evidence at time $T\geq \tau$. 
For each $\tau\in [0,\infty)$, define $\mathcal{T}(\tau):=\{T\in [\tau,\infty]:U(T;\tau)\geq U(T';\tau) \text{ for all } T'\in [\tau,\infty]\}$, and define $T^*(\tau)=\sup(\mathcal{T}(\tau))$. Since equilibrium requires a sender with evidence arriving at any $\tau$ to play a best response, $\mathcal{T}(\tau)$ must be nonempty for each $\tau$. If $T^*(\tau)=\tau$ for all $\tau$, then there is immediate disclosure almost surely at all times, so we have $S_t\equiv 0$ and we are done. Hence, assume that $T^*(\tau_0)>\tau_0$ for some $\tau_0$. If the set of such $\tau_0$ is bounded away from $0$, we are also done, because in this case we would again have $S_t=0$ on some interval $[0,\varepsilon)$ for $\varepsilon>0$. Hence, assume there exist such $\tau_0$ arbitrarily close to $0$. There are three possibilities: (i) $T^*(\tau_0)\in (\tau_0,\infty)$ for some $\tau_0$, (ii) $T^*(\tau_0)=\infty$ and $\infty\notin \mathcal{T}(\tau_0)$ for some $\tau_0$, and (iii) $\infty\in \mathcal{T}(\tau_0)$ for all $\tau_0$. Also note that since timestamps are payoff irrelevant for the sender, for any two timestamps $\tau<\tau'$ and times $T,T'\geq \tau'$, $U(T';\tau')-U(T;\tau')=U(T';\tau)-U(T;\tau)$, since flow payoffs prior to $\tau'$ are the same if disclosure occurs at $\tau'$ or later. It follows that for all $\tau<\tau'$, $\mathcal{T}(\tau)\cap [\tau',\infty]\subseteq \mathcal{T}(\tau')$. Moreover, if $\mathcal{T}(\tau)\cap[\tau',\infty]$ is nonempty, then $\mathcal{T}(\tau)\cap[\tau',\infty]=\mathcal{T}(\tau')$.

In case (i), for any date $\tau\in (\tau_0,T^*(\tau_0))$ at which good evidence arrives, since $T^*(\tau_0)$ is defined as a supremum, the set $[\tau,\infty]\cap \mathcal{T}(\tau_0)$ is nonempty and therefore is the same as $\mathcal{T}(\tau)$. The former set is bounded above by $T^*(\tau_0)$, and this implies that for all such $\tau$ there is no optimal disclosure time after $T^*(\tau_0)$. Now among good evidence types with arrival in $[0,\tau_0]$, if any do not disclose by time $\tau_0$ with probability one, they must disclose by time $T^*(\tau_0)$ because their preferences from time $\tau_0$ are the same as the $\tau_0$ type, who must disclose by $T^*(\tau_0)$. Moreover, there is zero probability of arrival at date $T^*(\tau_0)$ itself. So $S_{T^*(\tau_0)}=0$. 

We can rule out case (ii) as follows. In case (ii), there exists an increasing sequence $T_n$ in $\mathcal{T}(\tau_0)$ with $T_n\to \infty$. Now $T_n\to \infty$ implies $U(T_n;\tau_0)\to U(\infty;\tau_0)$. Also, $U(T_n;\tau_0)=u$ for some constant $u$. It follows that $U(\infty;\tau_0)=u$, and since each $T_n$ is a maximizer, $\infty$ must also be, contradicting the definition of case (ii).

In case (iii), toward deriving a contradiction, we argue that the sender must not disclose any bad evidence at any time.  Since the age of good evidence disclosed at time $t$ is at most $t$, we have $q^G_t\geq p^*+(1-p^*)e^{-\Lambda t}>p^*$. Let $F(q)=\frac{p^*}{r}+\frac{q-p^*}{\Lambda+r}$ denote the payoff from disclosing evidence (good or bad) that results in posterior $q$, and note that $F$ is strictly increasing. It follows that disclosing good evidence at time $t$ gives a payoff of $F(q^G_t)>F(p^*)$, and since the sender is willing to wait, it must be that $U(\infty;t)\geq F(q^G_t)$, giving $U(\infty;t)>F(p^*)$. Meanwhile, any bad evidence disclosure generates a posterior of at best $p^*(1-e^{-\Lambda t})\leq p^*$, so the payoff from disclosure of bad evidence is at most $F(p^*)$. It follows that the sender would never disclose bad evidence. 

Under any strategy in which bad evidence is never disclosed, we have $p_t\leq \phi_t$ for all $t>0$. But for all $\tau\in[0,t]$, $q^{G,\tau}_t>\phi_t$. Hence, after disclosing good evidence, even under the belief that it is from time $0$, the sender's resulting flow payoff is at least $q^{G,0}_t>\phi_t\geq p_t$ for all $t>0$. It follows that at all times, disclosing good evidence immediately would be strictly preferred to never disclosing, a contradiction.

We prove the second part of the lemma in the Online Appendix.

\subsection{Proof of Theorem \ref{thm:three_phase_eqbm}}
\subsubsection{Setup and Laws of Motion}
We use the notation and laws of motion introduced in Section \ref{subsec:construction}. In particular, $G_i$ denotes the joint probability that there is undisclosed good evidence and the time-$t$ state is $i$; $S_t=G_{0,t}+G_{1,t}$ denotes the total stock of undisclosed good evidence; $q_t=G_{1,t}/S_t$ is the receiver’s belief after disclosure of good evidence; and $p_t$ is the no-disclosure belief. 

During either the stockpiling or purging phase, $(G_1,G_0,p)$ have the following laws of motion \eqref{eq:G1_ODE_mix}-\eqref{eq:p_ODE_mix}. Change variables to $S=G_0+G_1$ and $q=G_1/S$ and differentiate to obtain 
\begin{align}
    \dot{S}_t &=(p_t-S_t q_t)\mu-S_t \kappa_t(1-S_t)\label{eq:S_ODE}\\
    \dot{q}_t &=-\lambda_1 q_t +(1-q_t)\mu\frac{p_t-S_t q_t}{S_t}\label{eq:q_ODE_chain_rule}
\end{align}

During the transparency phase, we have $q_t\equiv 1$ and $p$ evolves according to
 \begin{align}
     \dot{p}_t=-\lambda_1 p_t-\mu p_t(1-p_t).\label{eq:p_law_transparency}
\end{align}
Note that $p$ is decreasing during the transparency phase.

\subsubsection{Incentive compatibility}
Suppose the sender possesses good evidence at time $t_0$ and plans to disclose it at time $t\geq t_0$. Then for $s\geq t$, we have $\pi_s=q_t e^{-\lambda_1(s-t)}$. The sender's continuation payoff is
\begin{align*}
    F(t;t_0):=\int_{t_0}^t e^{-r(s-t_0)}p_s\,ds+ \int_t^\infty e^{-r(s-t_0)}q_t e^{-\lambda_1 (s-t)}\,ds.
\end{align*}
Indifference between disclosing at $t_0$ and at time  $t_0+\varepsilon$ for all sufficiently small $\varepsilon$ requires $F_t(t;t_0)|_{t=t_0}=0$; equivalently,
\begin{align*}
    0&=p_{t_0}-q_{t_0}+\int_{t_0}^\infty (\lambda_1 q_{t_0}+\dot{q}_{t_0})e^{-r(s-t_0)}e^{-\lambda_1(s-t_0)}\,ds=p_{t_0}-q_{t_0}+(\lambda_1 q_{t_0}+\dot{q}_{t_0})\frac{1}{\lambda_1+r}.
\end{align*}
Rearranging and replacing $t_0$ with $t$ yields the local indifference condition 
\begin{align}
    \dot{q}_t=r q_t-(
    \lambda_1+r)p_t,\label{eq:local_indiff}
\end{align}
which must hold in a purging phase. In a stockpiling phase, we must have $\dot{q}_t\geq r q_t-(\lambda_1+r)p_t$. In a transparency phase, we must have $\dot{q}_t\leq r q_t-(\lambda_1+r)p_t$. 
These conditions are sufficient for global optimality since they ensure that $F(t;t_0)$ is nondecreasing during a stockpiling phase, constant in a purging phase, and nonincreasing in a transparency phase.

\subsubsection{Identifying $\pund$ and $\pbar$}
In this section we prove the following lemma. Let $\pbar:=\frac{r+\lambda_1/2}{\lambda_1+r}$ and $\pund:=\frac{r}{\lambda_1+r}$.
\begin{lemma}\label{lem:pund_and_pbar}
    Any three-phase equilibrium must start in the stockpiling phase if $p_0 > \pbar$, the transparency phase if $p_0 \leq \pund$, and the purging phase if $p_0\in (\pund,\pbar]$.
\end{lemma}
\begin{proof}
    By definition, any three-phase equilibrium must start in one of the three phases. We consider each possibility separately and derive necessary conditions on $p_0$. We then conclude the proof by observing that the conditions are mutually exclusive, so that $p_0$ determines the starting phase as in the proposition.

\noindent\textbf{Starting in stockpiling phase:}
We argue that if the game starts in the stockpiling phase, then a sender who obtains good evidence at time $0$ must be willing to wait an instant, and this implies $p_0>\pbar=\frac{r+\lambda_{1}/2}{\lambda_1+r}$. 

A necessary and sufficient condition to start in the stockpiling phase is that for some time $\underline T>0$, waiting until $\underline T$ to disclose must be weakly preferred to disclosing at time $t$ for a sender who obtains good evidence at time $t\in [0,\underline T)$. Hence, recalling \eqref{eq:local_indiff}, a necessary condition for starting in stockpiling is that for some $\underline T>0$, 
\begin{align}
    \dot{q}_{\underline T}\geq r q_{\underline T}-(\lambda_1+r)p_{\underline T}.
\end{align}

Using $\kappa_t\equiv 0$, the laws of motion for beliefs admit a (unique) closed-form solution:
\begin{align}
    G_{1,t}&=e^{-\lambda_1 t}p_0(1-e^{-\mu t})\label{eq:G1_delay_closed_form}\\
    G_{0,t}&=e^{-\lambda_1 t}p_0\left(\frac{\lambda_1 e^{-t\mu}+\mu e^{\lambda_1 t}}{\mu+\lambda_1}-1\right)\label{eq:G0_delay_closed_form}\\
    p_t &=e^{-\lambda_1 t}p_0\label{eq:p_delay_closed_form}\\
    q_t &=\frac{G_{1,t}}{G_{0,t}+G_{1,t}}=\frac{(e^{\mu t}-1)(\lambda_1+\mu)}{(e^{(\lambda_1+\mu)t}-1)\mu} \text{ for $t>0$}.\label{eq:q_delay_closed_form}
\end{align}
Recall that $q_0=1$; it is easy to calculate the right derivative $\dot{q}_0=-\frac{\lambda_1}{2}$. 
   
Using \eqref{eq:p_delay_closed_form}, we get the necessary condition 
\begin{align*}
   p_0 &\geq b(t):=\frac{e^{\lambda_1 t}(rq_t-\dot{q}_t)}{\lambda_1+r}
\end{align*}
for some $t>0$,  
where $q_t$ is given by  \eqref{eq:q_delay_closed_form}. Define $b(0):=\lim_{t\downarrow 0}b(t)=\frac{r\cdot q_0-\dot{q}_0}{\lambda_1+r} =\frac{r\cdot 1+\frac{\lambda_1}{2}}{\lambda_1+r}=\pbar.$ 
We claim that $b'(t)>0$ for all $t>0$. Direct calculation shows that
\begin{align*}
    b'(t)&=\frac{e^{\lambda_1 t}(\lambda_1+\mu)}{(e^{(\lambda_1+\mu)t}-1)^3(\lambda_1+r)\mu}P(t),\qquad \text{where}\\
    P(t)&:=
\mu(\lambda_1+\mu+r)\,e^{t(2\lambda_1+2\mu)} -(\lambda_1+\mu)(2\lambda_1+\mu+r)\,e^{t(\lambda_1+2\mu)}  \\
&+(2\lambda_1^2+3\lambda_1\mu+\lambda_1 r+\mu^2-\mu r)\,e^{t(\lambda_1+\mu)}  -(\lambda_1+\mu)(\mu-r)\,e^{t\mu} -\lambda_1 r.
\end{align*}
For each of two cases $\mu\geq r$ and $\mu<r$, the exponential polynomial $P(t)$ has three sign changes and thus at most three real roots counting multiplicity. It is also easy to check that there is a triple root at $t=0$. Hence, there are no other real roots. Also, since the leading term is positive, $P(t)$ is  positive for sufficiently large $t$. We conclude that for all $t>0$, we have $P(t)>0$ and therefore $b'(t)>0$. It follows that $p_0\geq b(t)$ for some $t>0$ only if $p_0 > b(0) =\pbar.$
This is exactly the inequality asserted in Lemma \ref{lem:pund_and_pbar}.

For later use, let us define $c(t)=\dot{q}_t-rq_t+(\lambda_1+r)p_t$. Assume $p_0>\pbar$. We show that $c$ crosses $0$ from above at the unique point $t$ where $p_0=b(t)$. We have 
\begin{align*}
    c(t)&=(\lambda_1+r)e^{-\lambda_1 t}(p_0-b(t))\\
    \implies c'(t)&=-\lambda_1 c(t)-b'(t)(\lambda_1+r)e^{-\lambda_1 t}.
\end{align*}
When $c(t)=0$, the right hand side has the same sign as $-b'(t)<0$, as desired. Moreover, it is easy to show that $\lim_{t\to \infty}b(t)=\frac{\lambda_1+\mu}{\mu}>1>p_0$. Hence, there is a unique crossing point, which we denote by $\Tund$.

\noindent\textbf{Starting in transparency phase:}
In this case $q\equiv 1$ and $\dot{q}\equiv 0$, so at time $0$ the left hand side of \eqref{eq:local_indiff} is $0$ and the right hand side is $r-(\lambda_1+r)p_0$. Now for immediate disclosure to be optimal, we must have $0\leq r-(\lambda_1+r)p_0$, or equivalently, $p_0 \leq \frac{r}{\lambda_1+r}=\pund$.

\noindent\textbf{Starting in purging phase:} 
The sender must be indifferent between disclosing good evidence at time $0$ and disclosing at time $t$ for all sufficiently small $t>0$. Therefore, the local IC constraint \eqref{eq:local_indiff} must hold for all $t\in (0,\varepsilon)$ for sufficiently small $\varepsilon>0$. 

At $t=0$, we have $q_0=1$, and $q_t\geq e^{-\lambda_1 t}$ since good evidence is no older than $t$ units of time. Hence $q$ is continuous at $0$. Also, $p$ is continuous at $0$. First, we have $p_t< \phi_t=p_0 e^{-\lambda_1 t}$ for all $t\in (0,\varepsilon)$ since $\kappa_t>0$ on $(0,\varepsilon)$. Second, we have $p_t \geq \tilde{p}_t$, where $\tilde{p}$ is the solution to $\dot{p}_t=-\lambda_1 p_t-\mu p_t(1-p_t)$ with initial value $p_0$, which is the no-disclosure belief assuming a strategy of immediate disclosure of good evidence. For later, also note that we have strict inequality $p_t>\tilde{p}_t$ for all $t\in (0,\varepsilon)$ for sufficiently small $\varepsilon>0$. Both bounds $\phi$ and $\tilde{p}$ are continuous and equal $p_0$ at time $0$, so $\lim_{t\downarrow 0}p_t=p_0$.  

Now by \eqref{eq:local_indiff}, we have $\lim_{t\downarrow 0}\dot{q}_t=r-(\lambda_1+r)p_0$; since $q$ is continuous at $0$, this implies $\dot{q}_{0+}=r-(\lambda_1+r)p_0$. Hence, $\dot{q}_{t}=r q_t-(\lambda_1+r)p_t>r e^{-\lambda_1 t}-(\lambda_1+r)p_0 e^{-\lambda_1 t}=e^{-\lambda_1 t}(\lambda_1+r)(\pund-p_0)$ for all $t\in (0,\varepsilon)$. We must have $p_0> \pund$, otherwise we would have $\dot{q}_t>0$ for all $t\in (0,\varepsilon)$, and given $q_0=1$, and continuity at $0$, this would imply $q_t>1$ for some $t$.

We now argue that $\dot{q}_{0+}\geq -\frac{\lambda_1}{2}$.
Let $X_t=t-\tau$ be the (possibly random) age of evidence disclosed at $t$. Then since $\dot{q}_{0+}$ exists by the preceding argument,
\begin{align*}
    \dot{q}_{0+}=\lim_{t\to 0}\frac{\mathbb{E}[e^{-\lambda_1 X_t}]-1}{t}=\lim_{t\to 0}\frac{-\lambda_1 \mathbb{E}[X_t]+O(t^2)}{t}=-\lambda_1 \lim_{t\to 0}\frac{\mathbb{E}[X_t]}{t}.
\end{align*}
Expanding and writing $\gamma_s:=\phi_s \mu e^{-\mu s}$ yields
\begin{align*}
    \lim_{t\to 0}\frac{\mathbb{E}[X_t]}{t} = \lim_{t\to 0}\frac{\int_0^t (t-s)\gamma_s e^{-\int_s^t \kappa_u \,du}\,ds}{t \int_0^t \gamma_s e^{-\int_s^t \kappa_u \,du}\,ds}.
\end{align*}
Since $\gamma$ is continuous and $\gamma_0>0$, for any $\varepsilon>0$, we have that for all sufficiently small $t$, $\gamma_s\in ((1-\varepsilon)\gamma_0,(1+\varepsilon)\gamma_0)$ for all $s\in [0,t]$. Hence, 
\begin{align*}
     \lim_{t\to 0}\frac{\int_0^t (t-s)\gamma_s e^{-\int_s^t \kappa_u \,du}\,ds}{t \int_0^t \gamma_s e^{-\int_s^t \kappa_u \,du}\,ds}
\leq \frac{1+\varepsilon}{1-\varepsilon} \lim_{t\to 0}\frac{\int_0^t (t-s)e^{-\int_s^t \kappa_u \,du}\,ds}{t \int_0^t  e^{-\int_s^t \kappa_u \,du}\,ds}
\end{align*}
Since the density $\frac{e^{-\int_s^t \kappa_u \,du}}{\int_0^t  e^{-\int_x^t \kappa_u \,du}\,dx}$ is nondecreasing in $s$, while the function $t-s$ is decreasing, the weighted average is bounded above by replacing the density with a uniform density. Thus,
\begin{align*}
    \lim_{t\to 0}\frac{\mathbb{E}[X_t]}{t}\leq \frac{1+\varepsilon}{1-\varepsilon}\lim_{t\to 0}\frac{\int_0^t (t-s)\,ds}{t \int_0^t \,ds}=\frac{1+\varepsilon}{1-\varepsilon}\lim_{t\to 0}\frac{\frac{t^2}{2}}{t^2}=\frac{1+\varepsilon}{1-\varepsilon}\frac{1}{2}.
\end{align*}
Since $\varepsilon$ is arbitrary, we conclude that $\dot{q}_{0+}\geq -\frac{\lambda_1}{2}$, and thus
\begin{align*}
    p_0=\frac{r-\dot{q}_{0+}}{\lambda_1+r}\leq \frac{r+\frac{\lambda_1}{2}}{\lambda_1+r}.
\end{align*}

\noindent\textbf{Summary}
We have shown that starting in the stockpiling phase requires $p_0 >\pbar$, the purging phase $p_0\in (\pund,\pbar]$, and the transparency phase $p_0\leq \pund$. Since a three-phase equilibrium must start in one of these three phases, the only possibility is stockpiling for $p_0>\pbar$, transparency for $p_0\leq \pund$, and purging for $p_0\in (\pund,\pbar]$.
\end{proof}

We now establish existence and uniqueness within the class of three-phase equilibria for $p_0$ in each of the three intervals. 

\subsubsection{Low $p_0$.}
Suppose $p_0\leq \frac{r}{\lambda_1+r}$ and that the receiver conjectures that bad evidence is never disclosed, while good evidence is disclosed immediately. 

Then $q_t\equiv 1$ and $p$ evolves according to \eqref{eq:p_law_transparency}. Note that because $p$ is strictly decreasing, if $p_0\leq \pund$ then this inequality remains true at all future times. Hence, immediate disclosure of evidence is optimal at all histories, including after off-path delay of disclosure.

\subsubsection{High $p_0$.} Suppose that $p_0>\frac{r+\lambda_1/2}{r+\lambda_1}$. Fix any three-phase equilibrium. Since it must start in the stockpiling phase, the phases are characterized by cutoff times $0<\Tund_0\leq\Tbar_0<\infty$. Moreover, the possibility $\Tund_0=\Tbar_0$ is ruled out by the no-atom argument of Lemma \ref{lem:no_good_bang_bang}, which doesn't require the support restriction since there is permanent transparency on $[\Tbar_0,\infty)$. Let $(\kappa_t)_{t\in [\Tund_0,\Tbar_0)}$ be the disclosure intensity during the purging phase. We derive necessary conditions that uniquely determine these objects and then verify that the sender's strategy is optimal. 

Let $(p,q,S)$ denote the belief/stock variables in the candidate equilibrium. It is useful to define auxiliary variables $(p^0,S^0,q^0)$ as the belief/stock variables under an alternative conjecture that good evidence and bad evidence are never disclosed (with $q^0$ determined by neutral belief updating). 

Since both conjectures prescribe no disclosure on $[0,\underline T_0)$, and the variables $(p_t,q_t,S_t)$ and $(p_t^0,q_t^0,S_t^0)$ are continuous at $\underline T_0$ (since there are no disclosure atoms at $\underline T_0$), we have
\[
(p_t,q_t,S_t)=(p_t^0,q_t^0,S_t^0)\qquad \text{for all }t\le \underline T_0.
\]
These have closed-form expressions given by \eqref{eq:G1_delay_closed_form}-\eqref{eq:q_delay_closed_form} where $q_0=1$ and $q_t=G_{1,t}/(G_{0,t}+G_{1,t})$ for $t>0$.

\begin{lemma}\label{lem:unique_T_und_neutral_beliefs}
    We have $\Tund_0=\Tund$, where $\Tund$ is the unique time at which local IC \eqref{eq:local_indiff} holds under the no-disclosure beliefs:
    \begin{align*}
        \dot{q}^0_t= r q^0_t -(\lambda_1+r) p^0_t.
    \end{align*}
\end{lemma}
\begin{proof}
In the purging phase, $q$ must satisfy local indifference,
\begin{equation}\label{eq:IC_purging_altbeliefs}
    \dot q_t = r q_t-(\lambda_1+r)p_t\qquad \text{for }t\in [\underline T_0,\overline T_0),
\end{equation}
where we use the right-derivative. Also, by \eqref{eq:q_ODE_chain_rule} at $t=\underline T_0$,
\begin{align}\label{eq:Bayes_cutoff1}
    \dot{q}_{\underline T_0}&=-\lambda_1 q_{\underline T_0} +(1-q_{\underline T_0})\mu\frac{p_{\underline T_0}-S_{\underline T_0} q_{\underline T_0}}{S_{\underline T_0}},\\
    \dot{q}^0_{\underline T_0}&=-\lambda_1 q^0_{\underline T_0} +(1-q^0_{\underline T_0})\mu\frac{p^0_{\underline T_0}-S^0_{\underline T_0} q^0_{\underline T_0}}{S^0_{\underline T_0}}.\label{eq:Bayes_cutoff2}
\end{align}
Since 
$(p_{\underline T_0},S_{\underline T_0},q_{\underline T_0})=(p^0_{\underline T_0},S^0_{\underline T_0},q^0_{\underline T_0})$, \eqref{eq:Bayes_cutoff1} and \eqref{eq:Bayes_cutoff2} imply $
\dot q_{\underline T_0}=\dot q^0_{\underline T_0}$. 
Combining this with \eqref{eq:IC_purging_altbeliefs} and again using $(p_{\underline T_0},S_{\underline T_0},q_{\underline T_0})=(p^0_{\underline T_0},S^0_{\underline T_0},q^0_{\underline T_0})$ gives
\begin{align}
\dot q^0_{\underline T_0} = r q^0_{\underline T_0}-(\lambda_1+r)p^0_{\underline T_0}.\label{eq:stockpile_end}
\end{align}
Recall that from the construction of $\pbar$ in the proof of Lemma \ref{lem:pund_and_pbar}, for $p_0>\pbar=b(0)$, there is a unique value of $\underline T_0>0$, denoted $\underline T$, at which $p_0=b(\underline T_0)$, which is equivalent to \eqref{eq:stockpile_end}. Hence, in any three-phase equilibrium, $\underline T_0=\underline T$; in other words, the stockpiling phase must have length $\underline T$.
\end{proof}

We now characterize the purging phase solution: $\kappa_t$ and $\overline T$. The initial values for $(p_{\Tund},G_{1,\Tund},G_{0,\Tund},q)$ are already given from the end of the stockpiling phase. 
In the purging phase, $(p,G_1,G_0,q,\kappa)$ must satisfy the local indifference condition \eqref{eq:local_indiff} and the laws of motion \eqref{eq:G1_ODE_mix}-\eqref{eq:p_ODE_mix}. In addition, $q$ must also satisfy, by definition, $q_t=G_{1,t}/(G_{1,t}+G_{0,t})$ at each $t$ in $[\Tund_0,\Tbar_0)$.\footnote{Since $\Tund_0>0$, $G_{1,\Tund_0}/(G_{1,\Tund_0}+G_{0,\Tund_0})$ is well defined.} 

This system first reduces to the system in $(p,q,S,\kappa)$ given by \eqref{eq:p_ODE_mix}, \eqref{eq:S_ODE}, \eqref{eq:q_ODE_chain_rule}, and \eqref{eq:local_indiff}. In the next two steps, we eliminate $S_t$ and $\kappa_t$.

Next, we solve for $S_t$ by using the right side of \eqref{eq:local_indiff} to replace $\dot{q}_t$ in \eqref{eq:q_ODE_chain_rule}:
\begin{align}
    \dot{q}_t&=-\lambda_1 q_t +(1-q_t)\mu\frac{p_t-S_t q_t}{S_t}=rq_t-(\lambda_1+r)p_t\notag \\
    \implies S_t&=S(p_t,q_t):=\frac{\mu p_t(1-q_t)}{q_t[\lambda_1+r+\mu(1-q_t)]-(\lambda_1+r)p_t}.\label{eq:S_of_p_q}
\end{align}
A sufficient condition for $S(p_t,q_t)$ to be well defined is that $p_t< q_t\leq 1$, which we will verify.  

To solve for $\kappa_t$,  differentiate \eqref{eq:S_of_p_q}, substitute in the expressions for $\dot{p}$ and $\dot{q}$ from \eqref{eq:p_ODE_mix} and \eqref{eq:local_indiff}, and equate the resulting expression to the right hand side of \eqref{eq:S_ODE}. Solving for $\kappa_t$ yields
\begin{align}
\kappa_t&=\kappa(p_t,q_t):=\frac{A_t(q_t[\lambda_1+r+\mu(1-q_t)]-(\lambda_1+r)p_t)}{(q_t-p_t)(1-q_t)[\lambda_1+r+\mu(1-q_t)]}\label{eq:kappa_of_p_q}, \qquad \text{where}\\
A_t&:=A(p_t,q_t):=r+\lambda_1+\mu+(r-\mu)q_t-2p_t(\lambda_1+r).\label{eq:A_of_p_q}
\end{align}
As long as $0<p_t<q_t<1$, which we will verify in our solution, the denominator will be positive, so $\kappa(p_t,q_t)$ will be well defined. We will also need to verify that $\kappa_t> 0$ at all times. Since $0<p_t<q_t<1$ implies the second factor in the numerator is positive, it will be sufficient to show that $A_t\geq 0$ for all $t\in [\Tund,\Tbar)$.

After eliminating $S_t$ and $\kappa_t$ in this way, we obtain a reduced IVP consisting of the ODEs
\begin{align}
\dot{p}_t&=p_t\frac{2(\lambda_1+r)\mu p_t+\mu q_t(\lambda_1+\mu-r)-(\lambda_1+\mu)(\lambda_1+r+\mu)}{\lambda_1+r+\mu(1-q_t)}\label{eq:mix_IVP_p}\\
    \dot{q}_t &=rq_t-(\lambda_1+r)p_t\label{eq:mix_IVP_q},
\end{align}
with initial conditions $(p_{\Tund},q_{\Tund})$ from the end of the stockpiling phase.

Note that the right hand sides of \eqref{eq:mix_IVP_p} and \eqref{eq:mix_IVP_q} are time-invariant, $C^1$ functions of $(p_t,q_t)$ on the domain $U:=\{(p,q)\in \mathbb{R}^2:q<1+\frac{\lambda_1+r}{\mu}\}$. In particular, they are locally Lipschitz continuous, uniformly in time. By the Picard-Lindel{\"o}f theorem, there exists a unique solution to the IVP with initial conditions on some interval $[\Tund,\Tund+\varepsilon]$ with $\varepsilon>0$. Let $[\underline T,T^{\max})$ denote the maximal interval of existence (and uniqueness), where $T^{\max}$ is possibly $\infty$. 

Define $R\subset U$ by
\begin{align*}
    R:=\{(p,q):0<p<q<1 \text{ and } A(p,q)> 0\}.
\end{align*}
Then on $R$, (i) $\kappa_t$ is well defined via \eqref{eq:kappa_of_p_q} and strictly positive, and (ii) $S_t$ is well defined via \eqref{eq:S_of_p_q}. Define the (possibly infinite) time $T:=\inf\{t\in [\Tund,T^{\max}):(p_t,q_t)\notin R\}.$ 

\begin{lemma}\label{lem:purging_starts_in_R}
    We have that $(p_{\Tund},q_{\Tund})$ lies in the interior of $R$. Hence, $T > \Tund$.
\end{lemma}

For the proof (and other omitted proofs), see the Online Appendix. The following lemma establishes that if the solution exits $R$ at finite time $T$, it does so through $q_T=1$. 
\begin{lemma}
If $T$ is finite, then $q_T>p_T>0$ and $A_T> 0$. Hence, $q_T=1$.\label{lem:purging_exit_through_q=1}
\end{lemma}

We now show that $T$ is finite, so Lemma \ref{lem:purging_exit_through_q=1} indeed applies.

If $T=\infty$, the only possibility is that $T^{\max}=\infty$; otherwise, the solution would have to exit every compact subset of $U$ in finite time. Since $R\subset [0,1]^2\subset U$, this contradicts  $T=\infty$. The following lemma rules out this possibility. 
\begin{lemma}\label{lem:exit_time_finite}
It does not hold that $T=T^{\max}=\infty$.
\end{lemma}
\begin{proof}
Suppose $T=T^{\max}=\infty$. Then a unique solution exists for all time and remains in $R$. But the ODE \eqref{eq:mix_IVP_delta} for $\delta=q-p$ has solution
\begin{align*}
\delta_t=\delta_{\underline T} e^{\int_{\underline T}^t(r+S_s\kappa_s)\,ds}\geq \delta_{\underline T}e^{r(t-\underline T)}
\end{align*} 
where we have used that $S_t>0$ and $\kappa_t>0$ for all $t$, since  $(p_t,q_t)\in R$. In particular, $\delta$ reaches $1$ in finite time. This contradicts the fact that in $R$, $0<p_t<q_t<1$.
\end{proof}

To summarize, we have $T<\infty$, and the unique maximal solution exits $R$ at time $T$ with $q_T=1$. We define $\overline T=T$. It follows that as $t\uparrow \overline T$, $\kappa_t\to \infty$ and $S_t\to 0$. We then recover $G_1=q S$ and $G_0=S-G_1$, and $(G_1,G_0,p,\kappa)$ solve the original filtering equations \eqref{eq:G1_ODE_mix}-\eqref{eq:p_ODE_mix}. Note that with $\kappa_t>0$, standard comparison theorem arguments show that the belief/stock variables are feasible throughout the purging phase. To recover the strategy in CDF form, for $t\leq s<\Tbar$, we have $H(s;G,\tau,t)=1-e^{-\int_t^s\kappa_u\,du}$ where $\kappa_u=0$ for $u\in [0,\Tund)$, and for $s\geq \Tbar$ $H(s;G,\tau,t)=1$.

Because $q_t$ hits $1$ from below at $\overline T$, we have $\dot{q}_{\overline T-}\geq 0$. Moreover, we have $\dot{q}_{\overline T-}> 0$. To see this, suppose $\dot{q}_{\overline T-}=0$, which implies $p_{\overline T}=\frac{r}{\lambda_1+r}$. Evaluating the $p$-ODE at $(p_{\overline T},q_{\overline T})=(\frac{r}{\lambda_1+r},1)$ yields $\dot{p}_{\overline T-}<0$. But then using $q_{\overline T}=1$ and $\dot{q}_{\overline T-}=0$ yields $\ddot{q}_{\overline T-}=-(\lambda_1+r)\dot{p}_{\overline T-}>0$, which would imply $q_t>1$ for sufficiently large $t<\overline T$, a contradiction. Given that $\dot{q}_{\overline T-}>0$, we have $p_{\overline T}<\frac{r}{\lambda_1+r}$. From time $\overline T$ onward, we have $q_t=1$, and optimality of the sender's strategy follows exactly as before.

\subsubsection{Moderate $p_0$}
Suppose $p_0 \in (\pund,\pbar]$. In this case, the equilibrium must begin in the purging phase. The arguments are almost identical to those for high $p_0$, except that technical care is needed at the very first instant since $q_0=1$, so $\kappa(p_0,q_0)$ need not be well defined at time $0$. We address these details in the Online Appendix.

\subsection{Preliminaries for Section \ref{subsec:no_timestamps_bad}}
We begin with preliminary definitions. We define $q^B_t$ as the belief immediately after disclosure of bad evidence. Recall that $q^{B,s}_t$ is the belief at time $t$ conditional on bad evidence arriving at time $s$. 
Let $\sigma^G$ and $\sigma^B$ denote the time at which good or bad evidence, respectively, is disclosed (with the restriction that at most one of these is finite (and the other $+\infty$), since only one piece of evidence can arrive); we have $\sigma=\min\{\sigma^G,\sigma^B\}$. 

We define the following stock variables conditional on no disclosure up to time $t$:
\begin{align*}
    B_{i,t}&=\Pr(\tau \leq t \text{ and } \theta_\tau = 0 \text{ and } \theta_t=i | \sigma>t) \qquad \text{for $i\in \{0,1\}$},\\
    B_t&=\Pr(\tau \leq t \text{ and } \theta_\tau = 0 | \sigma>t)=B_{0,t}+B_{1,t}\\
    \bar{q}^B_t&=\begin{cases} \frac{B_{1,t-}}{B_{t-}}\qquad &\text{if }B_t>0\\
    0 \qquad &\text{if }B_t=0.
    \end{cases}
\end{align*}
It is also sometimes convenient to work with the unnormalized (i.e., not dividing by the probability of nondisclosure) stock variables. We define
\begin{align*}
    \Bun_{i,t}&=\Pr(\tau \leq t \text{ and } \theta_\tau = 0 \text{ and } \theta_t=i \text{ and } \sigma>t)\\
    \Bun_t&=\Pr(\tau \leq t \text{ and } \theta_\tau = 0 \text{ and } \sigma>t)\\
    \Nun_t &= \Pr(\tau > t).
\end{align*}

Before proving Proposition \ref{prop:no_ts_long_delay_bad}, we state and prove the following supporting lemma, which provides a useful lower bound on the evolution of $p_t$. Intuitively,  under age-independent bad evidence disclosure, such disclosure can only lower the belief, so nondisclosure of bad evidence can only raise it. The only remaining sources of downward pressure on the nondisclosure belief are therefore (i) Markov switching from state $1$ to $0$, and (ii) the fact that under immediate good evidence disclosure, no news is bad news. Since the proof is unrelated to the sender's optimization, we defer it to the Online Appendix.
\begin{lemma}\label{lem:long_delay_drift_bound}
Suppose that good evidence is disclosed immediately. Then under any age-independent bad evidence disclosure strategy, for all $t,s$ we have $p_{t+s}\geq p_t-(\lambda_1+\mu)s$. By implication, $p_t\geq p_{t-}$. 
\end{lemma}

\subsection{Proof of Proposition \ref{prop:no_ts_long_delay_bad}}
The time-$t$ payoff from disclosing bad evidence is 
    \begin{align}\label{eq:long_delay_disclose_now}
    \int_0^\infty e^{-rs}[p^*-(p^*-q^B_t)e^{-\Lambda s}]\,ds=\frac{p^*}{r}-\frac{p^*-q^B_t}{r+\Lambda}.
\end{align}
The payoff of disclosing at time $t+\varepsilon$ for any $\varepsilon>0$ is at least
\begin{align}\label{eq:long_delay_disclose_wait}
    \int_0^\varepsilon e^{-rs}p_{t+s}\,ds,
\end{align}
since the continuation flow utility after $t+\varepsilon$ is nonnegative. 

After subtracting \eqref{eq:long_delay_disclose_now} from \eqref{eq:long_delay_disclose_wait} and scaling by $r$, the benefit of delay by $\varepsilon$ is at least
\begin{align}\label{eq:delay_benefit_scaled}
    \int_0^\varepsilon r e^{-rs}p_{t+s}\,ds -p^*+\frac{r}{r+\Lambda}(p^*-q^B_t).
\end{align}
By Lemma \ref{lem:long_delay_drift_bound}, this is at least
\begin{align*}
    \int_0^\varepsilon r e^{-rs}[p_t-(\lambda_1+\mu)s]\,ds -p^*+\frac{r}{r+\Lambda}(p^*-q^B_t).
\end{align*}
As $r\to\infty$, this converges uniformly over $t$ and bad evidence disclosure strategies to 
$p_t-p^*+(p^*-q^B_t)=p_t-q^B_t.$ 
All that remains is to show that at all on-path bad evidence disclosure times $t\in [0,T]$, $p_t-q^B_t$ is uniformly bounded away from $0$.

In an equilibrium with age-independent strategies, if disclosure is on path at time $t$, then $q^B_t=\bar{q}^B_{t-}=\frac{\Bun_{1,t-}}{\Bun_{t-}}$ at all times, with $q^B_0=0$. We also have $p_t\geq p_{t-}$ by Lemma \ref{lem:long_delay_drift_bound}, where
    \begin{align}
        p_{t-}=\frac{\Nun_t \phi_t+\Bun_{1,t-}}{\Nun_t +\Bun_{t-}}.
    \end{align}
    Hence, 
    \begin{align*}
        p_t-q^B_t=p_t-\bar{q}^B_{t-}\geq p_{t-}-\bar{q}^B_{t-}
        &=\frac{\Nun_t}{\Nun_t +\Bun_{t-}}(\phi_t-\bar{q}^B_{t-})\geq \Nun_t(\phi_t-\bar{q}^B_{t-})=e^{-\mu t}(\phi_t-\bar{q}^B_{t-}).
    \end{align*}
We also have $\bar{q}^B_{t-}\leq q^{B,0}_t=p^*(1-e^{-\Lambda t})$ and $\phi_t=p^*-(p^*-p_0)e^{-\Lambda t}$, so
$\phi_t-\bar{q}^B_{t-}\geq p_0 e^{-\Lambda t}.$
Together, for all $t\in [0,T]$, these inequalities imply
\begin{align*}
    p_t-q^B_t \geq e^{-\mu t}p_0 e^{-\Lambda t}\geq p_0 e^{-(\Lambda+\mu)T}>0.
\end{align*}
Hence, by the uniform convergence before, for all sufficiently large $r$, the scaled benefit of delay by $\varepsilon$ is uniformly strictly positive over all $t\in [0,T]$ and all age-independent bad evidence strategies. Thus, delay is strictly optimal at every on-path bad evidence disclosure time in $[0,T]$.

\subsection{Proof of Lemma \ref{lem:bad_eventually_disclosed}}
We prove the second part of the lemma first. Specifically, we prove the following lemma.

\begin{lemma}\label{lem:no_timestamps_suff_cond_no_bad_disc}
    Assume $p_0\in \left(\frac{\lambda_0}{\Lambda+r},\frac{\lambda_0+r}{\Lambda+r}\right)$ and that $\mu<\left(\frac{\Lambda+r}{\lambda_1+r}\right)r$. Then there exists an equilibrium with immediate good evidence disclosure and no bad evidence disclosure.
\end{lemma}
\begin{proof}
We verify optimality of the sender's strategy assuming the receiver correctly conjectures immediate disclosure of good evidence and no disclosure of bad evidence, and that following off-path disclosure of bad evidence at time $t$, the receiver assumes the evidence arrived at time $t$.

Let $q^B_t$ (identically $0$ under our off-path belief assumption) denote the receiver's posterior immediately after disclosure of bad evidence.

Delaying bad evidence disclosure is locally incentive compatible if and only if
\begin{align*}
     \dot{q}^B_t\geq (\Lambda+r)(q^B_t-p_t)+\Lambda(p^{*}-q^B_t).
\end{align*}
With $q^B\equiv 0$, this reduces to
\begin{align}
    p_t \geq \frac{\lambda_0}{\Lambda+r},\label{eq:q:bound_no_bad_disc}
\end{align}
and hence to establish that it is optimal to never disclose bad evidence, it suffices to show that $p_t$ always satisfies \eqref{eq:q:bound_no_bad_disc}. By assumption, it is satisfied strictly at time $0$.

Let $q^G_t$ (identically $1$ under the conjectured strategy) denote the receiver's belief immediately after disclosure of good evidence. Delaying disclosure of good evidence is not locally incentive compatible if
\begin{align*}
     \dot{q}^G_t&<(\Lambda+r)(q^G_t-p_t)+\Lambda(p^{*}-q^G_t)=r q^G_t-(\Lambda+r)p_t+\lambda_0.
\end{align*}
With $q^G\equiv 1$, this reduces to
\begin{align}
    p_t <\frac{\lambda_0+r}{\Lambda+r}.\label{eq:bound_disclose_good_imm}
\end{align}
Hence, whenever \eqref{eq:bound_disclose_good_imm} holds, the sender prefers to disclose good evidence immediately rather than at any future time. By assumption, \eqref{eq:bound_disclose_good_imm} is satisfied at time $0$.

We now show that under the condition on $\mu$ in the proposition, $p_t \in \left( \frac{\lambda_0}{\Lambda+r},\frac{\lambda_0+r}{\Lambda+r}\right)$ for all $t\geq 0$ conditional on no disclosure. 

Let $B_{x,t}$ denote the probability that the sender has bad evidence and $\theta_t=x\in \{0,1\}$, conditional on no disclosure yet. The laws of motion for $(B_1,B_0,p)$ under the receiver's conjecture are
\begin{align}
    \dot{B}_{1,t}&=B_{0,t}\lambda_0-B_{1,t}\lambda_1+B_{1,t}(p_t-B_{1,t})\mu.\label{eq:B1_ODE_immediate_disc}\\
    \dot{B}_{0,t}&=-B_{0,t}\lambda_0+B_{1,t}\lambda_1+(1-p_t-B_{0,t})\mu+B_{0,t} (p_t-B_{1,t})\mu\label{eq:B0_ODE_immediate_disc}\\
    \dot{p}_t&=\lambda_0(1-p_t)-\lambda_1 p_t - \mu (p_t-B_{1,t})(1-p_t).\label{eq:p_ODE_immediate_disc}
    \end{align}
    with $B_{0,0}=B_{1,0}=0$ and $p_0\in (0,1)$ exogenously given.
Standard arguments using the comparison theorem show that there is a unique solution $(B_1,B_0,p)$ to this system for all time, and for $t>0$, all variables lie in $(0,1)$ and $p_t>B_{1,t}$. By inspection, the right hand side of \eqref{eq:p_ODE_immediate_disc} is negative when $p_t$ exceeds $p^*=\frac{\lambda_0}{\lambda_0+\lambda_1}$, which lies in $\left( \frac{\lambda_0}{\Lambda+r},\frac{\lambda_0+r}{\Lambda+r}\right)$. Hence, $p_t$ cannot hit the upper end of this interval in finite time. 

Next, note that the right hand side of \eqref{eq:p_ODE_immediate_disc} is bounded below by
\begin{align*}
    f(p_t):=\lambda_0(1-p_t)-\lambda_1 p_t - \mu (p_t-0)(1-p_t).
\end{align*}
Let $\tilde p$ denote the solution to $\dot{\tilde p}_t=f(\tilde p_t)$ with initial condition $\tilde p_0=p_0$. It is easy to show that $p_t\geq \tilde p_t$ at all times, and that $\tilde p_t$ converges monotonically to a steady state $\tilde p^*\in (0,1).$ The condition $\mu<\left(\frac{\lambda_0+\lambda_1+r}{\lambda_1+r}\right)r$ ensures that this steady state is above $\frac{\lambda_0}{\lambda_0+\lambda_1+r}$. Therefore, $p_t$ does not cross this threshold.
\end{proof}

For the first part of Lemma \ref{lem:bad_eventually_disclosed}, suppose the receiver conjectures immediate disclosure of good evidence and no disclosure of bad evidence. We state the explicit formulas for belief stock variables, which we refer to in later proofs as well:
\begin{align}
    B_{1,t}&=\frac{\int_0^t \mu e^{-\mu s}(1-\phi_s)q^{B,s}_t\,ds}{e^{-\mu t}+\int_0^t \mu e^{-\mu s}(1-\phi_s)\,ds}\label{eq:B1_closed}\\
    B_{0,t}&=\frac{\int_0^t \mu e^{-\mu s}(1-\phi_s)(1-q^{B,s}_t)\,ds}{e^{-\mu t}+\int_0^t \mu e^{-\mu s}(1-\phi_s)\,ds}\label{eq:B0_closed}\\
        \Bun_{1,t}&=\int_0^t \mu e^{-\mu s}(1-\phi_s)q^{B,s}_t \,ds\label{eq:B1un_closed}\\
    \Bun_{0,t}&=\int_0^t \mu e^{-\mu s}(1-\phi_s)(1-q^{B,s}_t) \,ds\label{eq:B0un_closed}\\
    \Bun_t&=\int_0^t \mu e^{-\mu s}(1-\phi_s) \,ds\label{eq:Bun_closed}\\
    \Nun_t &=e^{-\mu t}\label{eq:Nun_closed}.
\end{align}

Under neutral beliefs, the belief immediately after disclosure of bad evidence at time $t>0$ is $q^B_t=\bar{q}^B_t=\frac{B_{1,t}}{B_{0,t}+B_{1,t}}$, where $B_{1,t}$ and $B_{0,t}$ are given by \eqref{eq:B1_closed} and \eqref{eq:B0_closed}. We  also have 
\begin{align}
    p_t=\frac{\phi_t \Nun_t+\Bun_{1,t}}{\Nun_t+\Bun_{t}}\label{eq:p_identity}\\
    \bar{q}^B_t=\frac{B_{1,t}}{B_t}=\frac{\Bun_{1,t}}{\Bun_{t}}\label{eq:qbar_identity}.
\end{align}
A useful identity is
\begin{align}
    p_t-\bar{q}^B_t=\frac{\Nun_t}{\Nun_t+\Bun_t}(\phi_t-\bar{q}^B_t).\label{eq:p_barq^B_diff}
\end{align}
Note that this is positive, since $\bar{q}^B_t\leq q^{B,0}_t< \phi_t$. 
Moreover, $(\Bun_1,\Bun)$ satisfy the system
\begin{align}
    \dBun_{1,t}&=\lambda_0 (\Bun_t-\Bun_{1,t})-\lambda_1 \Bun_{1,t}\label{eq:dBun_1}\\
    \dBun_t&=\mu e^{-\mu t}(1-\phi_t)\label{eq:dBun}.
\end{align}
Thus $\bar{q}^B$ satisfies
\begin{align}
    \dot{\bar q}^B_t&=\frac{\dBun_{1,t}}{\Bun_t}-\bar{q}^B_t\frac{\dBun_t}{\Bun_t}=\lambda_0(1-\bar{q}^B_t)-\lambda_1 \bar{q}^B_t-\bar{q}^B_t\frac{\dBun_t}{\Bun_t}.\label{eq:dbarq^B_Bayes}
\end{align}

The local IC constraint for delaying disclosure is
\begin{align}\label{eq:local_IC_eventual_bad_disc}
    \dot{q}^B_t \geq  (\lambda_0+\lambda_1+r) (q^B_t-p_t)+(\lambda_0+\lambda_1)(p^*-q^B_t)=r q^B_t-(\Lambda+r)p_t+\lambda_0.
\end{align}
Note that if this condition fails for all $s\geq t$, then disclosure at $t$ is optimal. We now show that there exists such $t$. 

Using that $q^B_t=\bar{q}^B_t$ and  \eqref{eq:dbarq^B_Bayes}, \eqref{eq:local_IC_eventual_bad_disc} becomes
\begin{align}
  \lambda_0(1-\bar{q}^B_t)-\lambda_1 \bar{q}^B_t-\bar{q}^B_t\frac{\dBun_t}{\Bun_t}  &\geq r \bar{q}^B_t-(\Lambda+r)p_t+\lambda_0\notag\\
  \iff \bar{q}^B_t\frac{\dBun_t}{\Bun_t} &\leq (\Lambda+r)(p_t-\bar{q}^B_t).\label{eq:local_IC_eventual_bad_disc2}
\end{align}
By \eqref{eq:p_barq^B_diff} and rearranging, this is equivalent to
\begin{align}
   \Lambda+r\geq \bar{q}^B_t\frac{\dBun_t}{\Bun_t}\frac{\Nun_t+\Bun_t}{\Nun_t(\phi_t-\bar{q}^B_t)} = \bar{q}^B_t\frac{\mu(1-\phi_t)(\Nun_t+\Bun_t)}{\Bun_t(\phi_t-\bar{q}^B_t)}.\label{eq:bad_eventually_disclosed_limit}
\end{align}
On the right hand side, assuming $\lambda_1,\lambda_0>0$, we have the following limits. First, $\phi_t\to p^*\in (0,1)$, $\Nun_t\to 0$,  $\Bun_t\to \Bun_\infty:=\int_0^\infty \mu e^{-\mu s}(1-\phi_s)\,ds\in (0,1)$, and for each $s$, $\lim_{t\to \infty}q^{B,s}_t=p^*$. Next, by dominated convergence, $\Bun_{1,t} = \int_0^\infty \mathbbm{1}_{s\leq t} \mu e^{-\mu s}(1-\phi_s)q^{B,s}_t \,ds\to p^*\int_0^\infty \mu e^{-\mu s}(1-\phi_s) \,ds=p^*\Bun_\infty$. Hence, $\bar{q}^B_t\to p^*$, which implies $\phi_t-\bar{q}^B_t\to 0$ from above. 
Thus, the right-hand side of \eqref{eq:bad_eventually_disclosed_limit} diverges and the inequality fails for sufficiently large $t$, as desired.

\subsection{Proof of Lemma \ref{lem:no_bad_immediate}}
Toward a contradiction, suppose bad evidence is disclosed immediately at all $t\geq T$. Then all bad evidence disclosed at $t>T$ is assumed to be fresh: $q^B_t\equiv 0$ for all $t>T$. This implies that after disclosing bad evidence at any $t>T$, the sender's flow payoff at each $s\geq t$ is $\phi_{s-t}(0)$. Conditional on nondisclosure after $T$, the receiver believes that with probability $1$, the sender does not possess evidence of either type. Hence $p_t\equiv \phi_t$ for all $t>T$. It follows that by not disclosing bad evidence, the sender's flow payoff for $s\geq t$ is $\phi_s=\phi_{s-t}(\phi_t)>\phi_{s-t}(0)$. Thus, disclosing bad evidence at $t$ is not optimal for any $t>T$, a contradiction.

\subsection{Proof of Theorem \ref{thm:stationary_bursts}}
We show that such an equilibrium exists for large $r$, supported by neutral off-path beliefs after bad evidence disclosure strictly between bursts. 

We first establish optimality of immediately disclosing good evidence. Immediately after a burst, the receiver believes that the sender does not possess evidence, and thus beliefs coincide with $\phi_t=p^*=p_0$ given the stationary prior. In between bursts, the possibility that the sender possesses undisclosed bad evidence means that $p_t<p^*$. Hence, at all times we have $p_t\leq p^*<\frac{\lambda_0+r}{\Lambda+r}$, and by the argument in the proof of Lemma \ref{lem:no_timestamps_suff_cond_no_bad_disc}, disclosing good evidence immediately is indeed optimal when it is conjectured.

We now show that for large $r$, there exists $\Delta$ such that the bad evidence disclosure policy is optimal when it is conjectured. 

Suppose the bursts are conjectured to be spaced $\Delta$ units apart for some $\Delta>0$. Immediately after any burst, normalizing the current time to $0$, the belief variables evolve exactly as in \eqref{eq:B1_closed}-\eqref{eq:Nun_closed} and \eqref{eq:p_identity}-\eqref{eq:qbar_identity}. 

The first optimality condition is that delay must be optimal between bursts. Under the neutral off-path beliefs, disclosing at time $t \in (0,\Delta)$ results in a posterior $\bar{q}^B_t$, or at the next burst, $t=\Delta$, results in a posterior of $\bar{q}^B_\Delta$. Hence, the local delay IC condition is still \eqref{eq:local_IC_eventual_bad_disc} or equivalently \eqref{eq:local_IC_eventual_bad_disc2}. At time $0$, we have $\dot{q}^B_0>0$ while the RHS of \eqref{eq:local_IC_eventual_bad_disc} reduces to $-(\lambda_0+\lambda_1+r)p^*+(\lambda_0+\lambda_1)p^*=-rp^*<0$. Hence, by continuity, delay is strictly optimal initially. Moreover, we have already seen that  \eqref{eq:local_IC_eventual_bad_disc2} fails for large $t$; define $\Dr$ as:
\begin{align*}
    \Dr:=\inf\{t>0: \text{\eqref{eq:local_IC_eventual_bad_disc2} fails}\}.
\end{align*}
By definition, delay is locally incentive compatible in between bursts at regular intervals of length $\Delta$ if and only if $\Delta\leq \Dr$. 

The second optimality condition is that disclosure must be optimal at each burst. Suppose the time is $n\Delta$ and the sender has bad evidence. By disclosing it immediately, the sender obtains payoff calculated earlier as 
\begin{align*}
    \Pi(q^B_{\Delta}):=\frac{p^*}{r}-\frac{p^*-q^B_\Delta}{r+\Lambda},
\end{align*}
where we have used that $q^B_{n\Delta}=q^B_{\Delta}$ by the cyclical property of beliefs.

By waiting until $(n+1)\Delta$ to disclose, the sender obtains
\begin{align*}
    \int_{n\Delta}^{(n+1)\Delta}e^{-r(s-n\Delta)}p_s\,ds +e^{-r\Delta}\Pi(q^B_{\Delta}). 
\end{align*}
Disclosure at time $\Delta$ is therefore optimal if and only if
\begin{align}
    \Gamma(\Delta):=\Pi(q^B_{\Delta})(1-e^{-r\Delta}) - \int_{0}^{\Delta}e^{-rs}p_s\,ds \geq 0,\label{eq:burst_participation}
\end{align}
where we have used that beliefs are cyclical with period $\Delta$. 
Since the sender's problem is identical at each burst, this condition ensures that if the sender has bad evidence at any burst date, he is willing to disclose it.

The condition \eqref{eq:burst_participation} fails for all sufficiently small $\Delta$. In particular, we have equality at $\Delta=0$ while the derivatives of the left and right hand sides with respect to $\Delta$ at $\Delta=0$ are $r\Pi$ and $p^*$, respectively, where $\Pi<p^*/r$. However, \eqref{eq:burst_participation} holds for all sufficiently large $\Delta$; taking $\Delta\to\infty$, we have $q^B_{\Delta}\to p^*$ and therefore $\Pi(q^B_\Delta)\to \Pi(p^*)= p^*/r$, while the value of concealing forever is $\int_0^\infty e^{-rs}p_s\,ds<\int_0^\infty e^{-rs}p^*=p^*/r$. 

The rest of the proof consists of showing that for sufficiently large $r$, $\Gamma(\Dr)>0$, so that our conjectured strategy profile with $\Delta=\Dr$ is an equilibrium.

Let $y_t:=p^*-p_t$. For later use, note that $y_0=0, \dot{y}_0=\mu p^*(1-p^*)$, and $\ddot{y}_0=\mu p^*(1-p^*)(\mu(2p^*-1)-\Lambda)$. Then $\Gamma(\Dr)$ rearranges to
\begin{align*}
    \int_0^{\Dr} e^{-rs} y_s\,ds-(1-e^{-r\Dr})\frac{p^*-\bar{q}^B_{\Dr}}{\Lambda+r}
\end{align*}

\begin{lemma}\label{lem:r_Gamma_limits}
    As $r\to\infty$, we have $\Dr\to+\infty$, $r(p^*-\bar{q}^B_{\Dr})\to \dot{y}_0$, and $r^2\Gamma(\Dr)\to 0$.
\end{lemma}

In light of Lemma \ref{lem:r_Gamma_limits}, to show eventual positivity of $\Gamma(\Dr)$, we show that
\begin{align*}
    0<\lim_{r\to\infty} r^3\Gamma(\Dr)=\lim_{r\to\infty} \left(r^3\int_0^{\Dr} e^{-rs} y_s\,ds-r^3(1-e^{-r\Dr})\frac{p^*-\bar{q}^B_{\Dr}}{\Lambda+r}\right).
\end{align*}
Since the limit $r^2 \Gamma(\Dr)$ had the form $\dot{y}_0-\dot{y}_0$, taking limits on each term alone would yield $\infty-\infty$; so instead, we recenter by $\dot{y}_0$:
\begin{align}
r^3\Gamma(\Dr)=\underbrace{r^3\left(\int_0^{\Dr} e^{-rs} y_s\,ds-\frac{\dot{y}_0}{r^2}\right)}_{:=X(r)}-\underbrace{r^3\left((1-e^{-r\Dr})\frac{p^*-\bar{q}^B_{\Dr}}{\Lambda+r}-\frac{\dot{y}_0}{r^2}\right)}_{=:Y(r)}.\label{eq:r^3Gamma}
\end{align}

\begin{lemma}\label{lem:r^3Gamma_split}
    As $r\to \infty$, $X(r)\to \ddot{y}_0=\dot{y}_0(\mu(2p^*-1)-\Lambda)$ and $Y(r)\to \dot{y}_0(\max\{\Lambda-\mu,0\}-2\Lambda-\mu(1-p^*))<\ddot{y}_0$. Hence, $\lim_{r\to \infty}r^3\Gamma(\Dr)>0$.
\end{lemma}
By Lemma \ref{lem:r^3Gamma_split}, $\Gamma(\Dr)>0$ for sufficiently large $r>0$, so the construction is complete.

\bibliographystyle{jpe} 
\bibliography{timestampsbib}

\end{document}